\documentclass[11pt]{article}

\usepackage[margin=1.3in]{geometry}
\usepackage[utf8]{inputenc}
\usepackage[UKenglish]{babel}

\usepackage{amsmath}
\usepackage{amsthm}
\usepackage{amsfonts}
\usepackage{amssymb}

\newtheorem{theorem}{Theorem}
\newtheorem{lemma}{Lemma}
\newtheorem{corollary}{Corollary}

\newtheorem{definition}{Definition}

\usepackage{microtype}

\usepackage{placeins}
\let\Oldsection\section
\renewcommand{\section}{\FloatBarrier\Oldsection}
\let\Oldsubsection\subsection
\renewcommand{\subsection}{\FloatBarrier\Oldsubsection}
\let\Oldsubsubsection\subsubsection
\renewcommand{\subsubsection}{\FloatBarrier\Oldsubsubsection}

\usepackage{todonotes}
\usepackage{enumitem}
\usepackage{hyperref}
\usepackage[capitalize]{cleveref}
\usepackage{autonum}
\DeclareMathAlphabet{\mathpzc}{OT1}{pzc}{m}{it}

\DeclareMathOperator{\poly}{poly}
\DeclareMathOperator{\sign}{sign}
\DeclareMathOperator{\asvec}{vec}
\allowdisplaybreaks

\usepackage{xcolor}
\usepackage{mdframed}

\newcommand{\eps}{\varepsilon}
\newcommand{\R}{\mathbb{R}}
\newcommand{\Z}{\mathbb{Z}}

\newcommand{\abs}[1]{\left\lvert #1\right\rvert}

\newcommand{\norm}[2]{\left\| #1 \right\|_{#2}}

\newcommand{\Prp}[1]{\Pr\!\left[#1 \right]}
\newcommand{\Prpcond}[2]{\Pr\!\left[#1 \; \middle| \; #2 \right]}
\newcommand{\Ep}[1]{\text{E}\!\left[#1 \right]}
\newcommand{\Epcond}[2]{\text{E}\!\left[#1 \; \middle| \; #2 \right]}

\renewcommand{\(}{\left(}
\renewcommand{\)}{\right)}

\newcommand{\sjlmp}[2]{$(#1, #2)$-\hyperref[defn:strong-jl-moment]{Strong JL Moment Property}}

\author{
Thomas D. Ahle\\ \small ITU\\ \small \texttt{thdy@itu.dk}
\and
Jakob B. T. Knudsen\\ \small University of Copenhagen\\ \small \texttt{jakn@di.ku.dk}
}

\begin{document}


\title{
   Almost Optimal Tensor Sketch
}
\maketitle

\begin{abstract}
We construct a matrix $M\in\R^{m\otimes d^c}$ with just $m=O(c\,\lambda\,\eps^{-2}\poly\log1/\eps\delta)$ rows, which preserves the norm $\|Mx\|_2=(1\pm\eps)\|x\|_2$ of all $x$ in any given $\lambda$ dimensional subspace of $\R^d$ with probability at least $1-\delta$.
This matrix can be applied to tensors $x^{(1)}\otimes\dots\otimes x^{(c)}\in\R^{d^c}$ in $O(c\, m \min\{d,m\})$ time -- hence the name ``Tensor Sketch''.
(Here $x\otimes y = \asvec(xy^T) = [x_1y_1, x_1y_2,\dots,x_1y_m,x_2y_1,\dots,x_ny_m]\in\R^{nm}$.)

This improves upon earlier Tensor Sketch constructions by Pagh and Pham~[TOCT 2013, SIGKDD 2013] and Avron et al.~[NIPS 2014] which require $m=\Omega(3^c\lambda^2\delta^{-1})$ rows for the same guarantees.
The factors of $\lambda$, $\eps^{-2}$ and $\log1/\delta$ can all be shown to be necessary making our sketch optimal up to log factors.

With another construction we get $\lambda$ times more rows $m=\tilde O(c\,\lambda^2\,\eps^{-2}(\log1/\delta)^3)$, but the matrix can be applied to any vector $x^{(1)}\otimes\dots\otimes x^{(c)}\in\R^{d^c}$ in just $\tilde O(c\, (d+m))$ time.
This matches the application time of Tensor Sketch while still improving the exponential dependencies in $c$ and $\log1/\delta$.

By reductions in Avron et al. this gives new state of the art algorithms for kernel methods in algorithms such as linear regression and PCA.

   Technically, we show two main lemmas:
   (1) For many Johnson Lindenstrauss (JL) constructions, if $Q,Q'\in R^{m\times d}$ are independent JL matrices, the element-wise product $Qx \circ Q'y$ equals $M(x\otimes y)$ for some $M\in\R^{m\times d^2}$ which is itself a JL matrix.
   (2) If $M^{(i)}\in\R^{m\times md}$ are independent JL matrices, then $M^{(1)}(x \otimes (M^{(2)}y \otimes \dots)) = M(x\otimes y\otimes \dots)$
   for some $M\in\R^{m\times d^c}$ which is itself a JL matrix.
   Combining these two results give an efficient sketch for tensors of any size.

   To analyze the two constructions, we give new concentration for products of many JL matrices, as well as a higher order version of Khintchine's inequality, related to the higher order Gaussian chaos analysis by Latała~[Annals of Probability 2006].
\end{abstract}

\newpage
\tableofcontents
\newpage


\section{Introduction}

A classical generalization of linear models in statistics are so called ``bilinear models''~\cite{lin2015bilinear}.
Instead of training a model on a vector $x = \langle x_1, x_2, \dots, x_d\rangle$ we do it on $x\otimes x = \langle x_1x_1, x_1x_2, \dots x_1x_d, \dots, x_dx_d\rangle$.
This allows learning relationships that depend non-linearly on the input.
The method extends to higher orders, e.g. $x\otimes x\otimes x = \langle x_1x_1x_1, \dots, x_dx_dx_d\rangle$, but it is clear that the produced vectors will quickly get too large to handle.
A natural question is whether we can compute a vector that behaves like $x\otimes x\otimes x$
, but has much smaller dimension?

Feature hashing~\cite{DBLP:conf/icml/WeinbergerDLSA09} and the Johnson Lindenstrauss transformation~\cite{johnson1984extensions} are the traditional solutions to this problem.
To get from dimension $d^3$, say, to some much smaller $m$, we apply a matrix $M\in\R^{m\times d^3}$ where $\langle M(x\otimes x\otimes x), y\rangle \approx \langle (x\otimes x\otimes x), y\rangle$ for any vector $y\in\R^m$.
The classic approaches to dimensionality reduction are however too slow, since they require the tensor product to be computed exact before they are applied.
A shortcut called ``Tensor Sketch'' was introduced in \cite{pham2013fast} which allows $M(x\otimes x\otimes x)$ to be computed with roughly $O(d+m)$ time and space.
Random weights applied to tensor products was tested in~\cite{gao2016compact} in combination with convolutional neural networks for classification.
The success follows a recent neural network trend in which a layer with random weights often works just as well as one trained for the task at hand~\cite{guest2017success}.
The issue with the approach of~\cite{pham2013fast} and~\cite{gao2016compact} however is that it grows exponentially in the number of tensorings, thus $x\otimes\dots\otimes x$ ($c$ times) takes time $\Omega(3^c)$ to compute with their approach.

A more rigorous approach was taken by Avron et al.~\cite{DBLP:conf/nips/AvronNW14} and Woodruff~\cite{DBLP:journals/fttcs/Woodruff14}.
They proved that a Tensor Sketch with sufficiently many rows is an oblivious ``subspace embedding'' and used this to give state of the art algorithms for problems such as kernel linear regression and kernel PCA.
For example, if we want to take a simple regression model like $\arg\min_{w\in\R^d} \|Xw-y\|_2$, where $X\in\R^{n\times d}$ and $y\in\R^d$, and turn it into a bilinear model, we make a new matrix $X'\in \R^{n\times d^2}$ such that $X'_{i} = X_i\otimes X_i=\langle X_{i,1}X_{i,1},  X_{i,1}X_{i,2},\dots, X_{i,n}X_{i,n}\rangle$, and solve $\arg\min_{w\in\R^{d^2}} \|X'w-y\|_2$.
To speed this up, we instead solve $\arg\min_{w\in\R^{m}} \|(M^T X')w-y\|_2$ for a sketch matrix $M$.
It can be shown that the optimal $w$ will always be in the subspace spanned by the rows of $X'$, so taking $M$ with sufficiently many rows that it preserves all vectors within any subspace of dimension $n$, the algorithm succeeds with error $\pm\eps$.
Taking advantage of the rule $\langle x\otimes x, y\otimes y\rangle = \langle x,y\rangle^2$ we can extend this method to any polynomial kernel function.
For any polynomial $\phi:\R^d\times\R^d\to\R^k$ there are functions $f,g\in\R^d\to\R^m$ such that $\langle f(x),g(y)\rangle = \phi(x,y)$ and using the above method $k$ can be taken to be $O(3^c \eps^{-2} n^2)$ where $c$ is the degree of $\phi$.
The problem with this method is that some classical kernel functions, such as the Gaussian radial bias function requires $c$ super constant to observe a polynomial approximation.
The quadratic dependency on the subspace dimension is also not ideal.
In this work we show how to get rid of both of these problems.

In some applications of kernel functions, the algorithms at hand only depend on the inner product between vectors and the ``kernel trick'' saves having to produce explicit representations $f(x)$ and $g(y)$.
However calculating the kernel function $\phi$ quickly becomes a bottleneck, in particular since the oracle-only nature of the kernel trick tends to require that it is computed for all pairs of data points.
There are other methods than tensor sketch for explicitly sketching vectors for kernels, as we review in our Related Work section, however since they are not linear transformations, they are not known to have the subspace embedding mentioned above, which makes it very hard to get guarantees of correctness.
Also note that while the original Tensor Sketch algorithm requires $\Omega(\delta^{-1})$ rows to get success probability $1-\delta$, one can instead take the median of $\log1/\delta$ sketches that each succeed with constant probability.
However taking the median of multiple sketches is also not a linear operation, and thus for the purpose of subspace embeddings all Tensor Sketch constructions previous to this paper have a huge blow up in the number of rows required to get just $1-1/n$ success probability.

\subsection{Technical Overview}

Say we want a matrix $M$ such that $\norm{M(x\otimes y)}2=(1\pm\eps)\norm{x\otimes y}2$
for any $x,y\in\R^d$, and we're willing to spend time most $\tilde O(d)$.
If $x,y$ where binary $\in\{-1,1\}^d$ then sampling elements from $x\otimes y$ would have good concentration and so $M$ could simply be a sampling matrix.
Generalizing this for $x,y\in\R^d$ we may limit the maximum value by first applying fast random random rotations $R,R'\in\R^{d\times d}$ and then sampling from $Rx\otimes R'y$.
By symmetry we might as well sample the diagonal $=Rx\circ R'y=[(Rx)_1(R'y)_1,(Rx)_2(R'y)_2,\dots,(Rx)_m(R'y)_m]^T$.
This has reasonably good concentration by e.g. Bernstein bounds.

What else can be done?
The Tensor Sketch of Pham and Pagh~\cite{pham2013fast} computes $C^{(1)}x \,\ast\, C^{(2)}y$,
where $C^{(1)}$ and $C^{(2)}$ are independent Count Sketch matrices and $\ast$ is vector convolution.
They show that, amazingly, this equals $C(x\otimes y)$ --- a count sketch of the tensor product!
This method, surprisingly, turns out the be equivalent to our ``first approach'', as
$C^{(1)}x \,\ast\, C^{(2)}y= \mathcal F^{-1}(\mathcal F C^{(1)} x \,\circ\, \mathcal F C^{(2)} y)$, where $\mathcal F$ is the Fourier transform.
Since $\mathcal F$ is an orthonormal matrix, $\mathcal F^{-1}$ doesn't impact the norm of $Cx$ and may be ignored.
What's left is that Tensor Sketch simply rotates each vector with a matrix $\mathcal F C^{(i)}$, similar to the Fast Johnson Lindenstrauss Transform by Ailon et al.~\cite{ailon2006approximate}, and takes the element wise product of the resulting vectors.

Given two matrices $M^{(1)}$ and $M^{(2)}$, the matrix $M$ such that $M(x\otimes y) = M^{(1)}x \circ M^{(2)}y$ is the row-wise tensor product where $M_{i} = M^{(1)}_i \otimes M^{(2)}_i$ which we write as $M=M^{(1)}\bullet M^{(2)}$.
From the above examples it seems like, given any good random rotations $M^{(1)}, M^{(2)}, \dots$ we can simply construct $M = M^{(1)}\bullet M^{(2)}\bullet\dots$ and get a good, fast Tensor Sketch.
Up to a few issues about independence of the rows of the $M^{(i)}$s we show that this is indeed the case!

By using known random rotations with strong probabilistic guarantees, we show that it suffices to have roughly $(\log1/\delta)^c$ rows for a $c$th-order Tensor Sketch.
This compares to the second-moment analysis of the Count Sketch based approach, which required $3^c\delta^{-1}$ rows.
For small $c$ this is an exponential improvement in $\log1/\delta$.
We would however also like to get rid of the exponential dependence in $c$.
Unfortunately we show a lower bound implying that this dependence is needed for any construction on the form above.

The second idea of this paper is to reduce this dependency to be linear in $c$, by a recursive construction as follows:
After computing a pair-wise sketch $Mx\circ M'y$ we map it back down to a more manageable size, such that $Mx \circ M' (M'' y \circ M''' z)$ becomes the sketch of a third order tensor etc.
By union bounding the error over $c$ such steps one easily gets down to a $c^2$ dependency in the number of rows.

Getting the $c$ dependence down from quadratic to linear takes a much more detailed analysis, which is the main focus of the second section.
A key part is the introduction of the following strengthening of the Johnson Lindenstrauss property for random matrices:
\begin{definition}[$(\eps,\delta)$-Strong JL Moment Property]\label{defn:strong-jl-moment}
   Let $\eps,\delta\in[0,1]$.
   We say a distribution over random matrices $M\in \R^{m\times d}$ has the $(\eps,\delta)$-Strong JL Moment Property, when
   $\Ep{\norm{Mx}2^2}=1$ and
   $\left(\Ep{\left(\norm{Mx}2^2-1\right)^p}\right)^{1/p} \le \frac{\eps}e\sqrt{\frac{p}{\log1/\delta}}$,
   for all unit vectors $x\in\R^d$ and all $p$ such that $2 \le p\le \log1/\delta$.
\end{definition}
The beginning of \cref{sec:techniques} shows more properties of the JL properties and how they relate.



\subsection{Theorems}

Our most general theorem concerns the recursive combination of matrices for higher order tensor sketching.
In stating it we assume that the output dimension, $m$ is between the vector dimension, $d$ and the tensored dimension, $d^c$.
If $d$ is larger than $m$, one may always start by reducing down to $m$ w.l.o.g.

\begin{theorem}[General Sketch]\label{thm:gen-sketch}
   Assume a distribution over matrices $M\in \R^{m\times dm}$ with the \sjlmp{\eps/\sqrt{c}}{\delta}, and where for any $x\in\R^m, y\in\R^d$, $M$ can be applied to $x\otimes y\in\R^{dm}$ in time $T$.
   Then there is a distribution over matrices $M'\in \R^{m\times d^c}$
   such that
   \begin{enumerate}
      \item $M'$ can be applied to vectors on the form $x^{(1)}\otimes\dots\otimes x^{(c)}\in \R^{d^c}$ in time $O(c\, T)$.
      \item $M'$ has the \sjlmp{\eps/\sqrt{c}}{\delta} over $\R^{d^c}$.
   \end{enumerate}
\end{theorem}

The theorem also gives results related to Oblivious Subspace Embeddings using \cref{lemma:ose} and \cref{lemma:ose2}.
The idea is to analyse a matrix with the property
$
   M(x\otimes y \otimes z \dots)
   =
   M^{(1)}(x \otimes M^{(2)}(y\otimes M^{(3)}(z \otimes \dots)))
$.
It turns out that it suffices to show that the Strong JL Moment property is preserved by matrix direct product and multiplication.

Combined with a theorem of Kraemer et al.~\cite{krahmer2011new} one may the ``Fast JL'' matrix of~\cite{johnson1984extensions} in \cref{thm:gen-sketch} to get the two of the results in \cref{tab:results}.
However the direct applications of the theorem doesn't use that we know how to efficiently sketch order 2 tensors.
While the number of rows are near optimal, the application time suffers.

To fix this we give two strong families of sketches for low order tensors.
For generality we prove them for general $c$-order tensor products $x^{(1)}\otimes\dots\otimes x^{(c)}$, but when used with \cref{thm:gen-sketch} we will just take $c=2$.
{
   \renewcommand{\thetheorem}{\ref{thm:indeprows}}
   \begin{theorem}
   Let $\eps,\delta\in[0,1]$ and let $c\ge 1$ be some integer.
   Let $T\in\R^{m\times d}$ be a matrix with iid. rows $T_1, \dots, T_m \in \R^d$ such that $\Ep{(T_1x)^2}=\|x\|_2^2$ and $\norm{T_1x}p \le \sqrt{a p} \|x\|_2$ for some $a>0$ and all $p\ge 2$.
   Let $M = T^{(1)}\bullet\dots\bullet T^{(c)}$ where $T^{(1)}, \dots, T^{(c)}$ are independent copies of $T$.
   Then there exists a constant $K>0$ such that $M$ has the \sjlmp{\eps}{\delta} given
   \begin{align}
      m \ge K \left[(4a)^{2c} \eps^{-2} \log1/\delta
      + (2ae)^{c} \eps^{-1} (\log1/\delta)^{c}\right]
      .
   \end{align}
\end{theorem}
   \addtocounter{theorem}{-1}
}
In the case where the $T^{(i)}$ are random Rademacher matrices and $c$ is constant, the bound simplifies to give that $m=O\left(\eps^{-2} \log1/\delta + \eps^{-1} (\log1/\delta)^c\right)$ rows suffices.
We will show in \cref{appendix:lowerbound-subgaussian} that this analysis is indeed tight.

While the family of \cref{thm:indeprows} takes advantage of the tensor structure, it is not fast on individual vectors, meaning we still haven't reached our goal of sketching $x^{\otimes c}\in\R^d$ in time near linear in $d$.
We do this in our final construction, which is a version of Fast JL specifically analysed on tensor products:
{
   \renewcommand{\thetheorem}{\ref{thm:tensorfastjl}}
   \begin{theorem}[Fast Construction]
      Let $T^{(1)}\in \R^{m\times d}$ be a Fast JL matrix, and let $T^{(2)}, \dots, T^{(c)}$ be independent copies,
      and define $M$ by $M_i = T^{(1)}_i \otimes \dots \otimes T^{(c)}_i$.
      Then taking one can take
      $m=O(K^c \eps^{-2}(\log1/\delta)(\log1/(\eps\delta))^c)$ for some constant $K$ to get
      \begin{enumerate}
         \item $M$ has the \sjlmp\eps\delta.
         \item $M$ can be applied to tensors $x^{(1)}\otimes\dots\otimes x^{(c)}\in \R^{d^c}$ in time $O(c\,(d\log d+m))$.
      \end{enumerate}
   \end{theorem}
   \addtocounter{theorem}{-1}
}

Combined with \cref{thm:gen-sketch} these two constructions give the remaining results in \cref{tab:results}.

\vspace{1em}

We mention a classical application of the above theorems:

\begin{corollary}[Polynomial Kernels]
   Let $\mathcal M$ be a distribution of matrices as in \cref{thm:gen-sketch}.
   Let $P : \R \to \R$ be a degree $c$ polynomial, then there is a linear map $M:\R^d\to\R^m$ such that $\langle Mx, My\rangle = P(\langle x,y\rangle)$.
   $M$ can be computed in time $cT$.
\end{corollary}
This follows simply from the rule $\langle x^{\otimes k}, y^{\otimes k}\rangle = \langle x,y\rangle^k$ combined with Horner's rule $a_0\oplus a_1 x \oplus a_2 x^{\otimes 2} \oplus \dots = a_0 \oplus x \otimes (a_1 \oplus x\otimes(a_2 \dots$.
One can extend this technique for general symmetric polynomials $P:\R^{2d}\to\R$
such that $\langle \phi(x), \psi(y)\rangle = P(x,y)$.

\subsection{Related work}

Work related to sketching of tensors and explicit kernel embeddings is found in fields ranging from pure mathematics to physics and machine learning.
Hence we only try to compare ourselves with the four most common types we have found.

We focus particularly on the work on subspace embeddings~\cite{pham2013fast, DBLP:conf/nips/AvronNW14}, since it is most directly comparable to ours.
An extra entry in this category is~\cite{shdpk}, which is currently in review, and which we were made aware of while writing this paper.
That work is in double blind review, but by the time of the final version of this paper, we should be able to cite it properly.

\paragraph*{Subspace embeddings}

For most applications~\cite{DBLP:conf/nips/AvronNW14}, the subspace dimension, $\lambda$, will be much larger than the input dimension, $d$, but smaller than the implicit dimension $d^c$.
Hence the size of the sketch, $m$, will also be assumed to satisfy $d \ll m \ll d^c$ for the purposes of stating the results.
We will hide constant factors, and $\log1/\epsilon$, $\log d$, $\log m$, $\log c$, $\log\lambda$ factors.

Note that we can always go from $m$ down to $\approx\epsilon^{-2}(\lambda+\log1/\delta)$ by applying an independent JL transformation after embedding.
This works because the product of two subspace embeddings is also a subspace embedding\footnote{
Let $S$ and $T$ both preserve subspaces of dimension $\lambda$, then given some such subspace, $T$ maps it to some other subspace of dimension at most $\lambda$ which is then preserved by $S$.},
and because standard JL is a subspace embedding by the net-argument (see lemma~\ref{lemma:ose2}).
The embedding dimensions in the table should thus mainly be seen as time and space dependencies, rather than the actual embedding dimension for applications.

\begin{table}[!htb]
\begin{center}
   \begin{tabular}{ l | l | l}
      Reference & Embedding dimension, $m$ & Embedding time
      \\ \hline
      \cite{pham2013fast, DBLP:conf/nips/AvronNW14}
                & $\tilde O(3^c\, \lambda^2\, \delta^{-1}\, \epsilon^{-2})$
                & $\tilde O(c\, (d + m))$
      \\ \hline
      Theorem~\ref{thm:gen-sketch}
                & $\tilde O(c\, \lambda^2\, (\log1/\delta)^3\, \epsilon^{-2})$
                & $\tilde O(c\, (d+m))$
      \\
      Theorem~\ref{thm:gen-sketch}
                & $\tilde O(c\, (\lambda + \log1/\delta) \eps^{-2} + c (\lambda + \log1/\delta)^2 \eps^{-1})$
                & $\tilde O(c\, d\, m)$
      \\
      Theorem~\ref{thm:gen-sketch}
                & $\tilde O(c\, \lambda^2 (\log1/\delta) \eps^{-2} + c \lambda (\log1/\delta)^2 \eps^{-1})$
                & $\tilde O(c\, d\, m)$
      \\ \hline
      Theorem~\ref{thm:gen-sketch} + \cite{johnson1984extensions}
                & $\tilde O(c\, (\lambda + \log 1/\delta) \eps^{-2})$
                & $\tilde O(c\, d\, m^2)$
      \\
      Theorem~\ref{thm:gen-sketch} + \cite{krahmer2011new}
                & $\tilde O(c\, \lambda\, (\log1/\delta)^{O(1)} \epsilon^{-2})$
                & $\tilde O(c\, d\,  m)$
      \\ \hline
      \cite{shdpk} $(\ast)$ Theorem 1
                & $\tilde O(c\, \lambda^2\, \delta^{-1}\, \epsilon^{-2})$
                & $\tilde O(c\,(d + m))$
      \\
      \cite{shdpk} $(\ast)$ Theorem 2
                & $\tilde O(c^6\, \lambda\, (\log1/\delta)^5\, \epsilon^{-2})$
                & $\tilde O(c\,(d + m))$
      \\ \hline
      \cite{sarlos_neurips} $(\ast), (\ast\ast)$ Theorem 2.1
                & $\tilde O(\eps^{-2}\log1/\eps\delta + \eps^{-1}(\log1/\eps\delta)^c)$
                & $\tilde O(c\,d\,m)$
      \\ \hline
   \end{tabular}
\end{center}
\caption{Comparison of embedding dimension and time for some of the results in the article and previous/concurrent work.}
\label{tab:results}
\end{table}
A few notes about the contents of \cref{tab:results}:
\begin{itemize}
   \item The results with $(\ast)$ are from unpublished manuscript communicated to us by the authors.
   \item Some of the results, in particular \cite{pham2013fast, DBLP:conf/nips/AvronNW14}, \cite{shdpk}~Theorem 1 and \cite{sarlos_neurips}~Theorem 2.1 can be applied faster when the input is sparse.
         Our results, as well as~\cite{shdpk}, Theorem 2 can similarly be optimized for sparse inputs, by preprocessing vectors with an implementation of Sparse JL~\cite{cohen2018simple}.
   \item The result $(\ast\ast)$ assumes $c=O(1)$ which is why it appears to have fewer dependencies on $c$.
\end{itemize}
In comparison to the previous result~\cite{pham2013fast, DBLP:conf/nips/AvronNW14} we are clearly better with an exponential improvement in $c$ as well as $\delta$.
Compared to the new work of~\cite{shdpk}, all four bounds have some region of superiority.
Their first bound of has the best dependency on $c$, but has an exponential dependency on $\log1/\delta$.
Their second bound has an only linear dependency on $d+\lambda$, but has large polynomial dependencies on $c$ and $\log1/\delta$.

Technically the methods of all five bounds are similar, but some details and much of the analysis differ.
Our results as well as the results of~\cite{shdpk} use recursive constructions to avoid exponential dependency on $c$, however the shape of the recursion differs.
We show all of our results using the $p$-moment method, while~\cite{shdpk} Theorem 1 and~\cite{pham2013fast, DBLP:conf/nips/AvronNW14} are shown using 2nd-moment analysis.
This explains much of why their dependency on $\delta$ is worse.

\paragraph*{Approximate Kernel Expansions}
A classic result by Rahimi and Rect~\cite{rahimi2008random}
shows how to compute an embedding for any shift-invariant kernel function $k(\|x-y\|_2)$ in time $O(dm)$.
In~\cite{DBLP:journals/corr/LeSS14} this is improved to any kernel on the form $k(\langle x,y\rangle)$ and time $O((m+d)\log d)$.
This is basically optimal in terms of time and space,
however the method does not handle kernel functions that can't be specified as a function of the inner product, and it doesn't provide subspace embeddings.
See also~\cite{musco2017recursive} for more approaches along the same line.

\paragraph*{Tensor Sparsification}
There is also a literature of tensor sparsification based on sampling~\cite{nguyen2015tensor}, however unless the vectors tensored are already very smooth (such as $\pm1$ vectors), the sampling has to be weighted by the data.
This means that these methods in aren't applicable in general to the types of problems we consider, where the tensor usually isn't known when the sketching function is sampled.

\paragraph*{Hyper-plane rounding}
An alternative approach is to use hyper-plane rounding to get vectors on the form $\pm1$.
Let $\rho = \frac{\langle x,y\rangle}{\|x\|\|y\|}$, then we have
$
\langle\sign(Mx), \sign(My)\rangle
= \sum_i \sign(M_i x)\sign(M_i y)
= \sum_i X_i
$
,
where $X_i$ are independent Rademachers with $\mu/m = E[X_i] = 1-\frac{2}{\pi}\arccos\rho = \frac{2}{\pi}\rho + O(\rho^3)$.
By tail bounds then
$\Pr[|\langle\sign(Mx), \sign(My)\rangle - \mu|>\epsilon\mu]
\le2\exp(-\min(\frac{\epsilon^2\mu^2}{2\sigma^2}, \frac{3\epsilon\mu}{2}))
$.
Taking $m = O(\rho^{-2}\epsilon^{-2}\log1/\delta)$ then suffices with high probability.
After this we can simply sample from the tensor product using simple sample bounds.

The sign-sketch was first brought into the field of data-analysis by~\cite{charikar2002similarity} and~\cite{valiant2015finding} was the first, in our knowledge, to use it with tensoring.
The main issue with this approach is that it isn't a linear sketch, which hinders some applications, like subspace embeddings.
It also takes $dm$ time to calculate $Mx$ and $My$.
In general we would like fast-matrix-multiplication type results.


\subsection{Notation and Preliminaries}

\paragraph{Asymptotic notation}
We say $f(x) \lesssim g(x)$ if $f(x) \le C g(x)$ for all $x\in\R$ and some universal constant $C$.
Note this is slightly different from the usual $f(x)=O(g(x))$ in that it is uniform rather than asymptotic.

For $p\ge 1$ and random variables $X\in R$, we write
$\|X\|_p = (E |X|^p)^{1/p}$.
Note that $\|X+Y\|_p\le\|X\|_p+\|Y\|_p$ by the Minkowski Inequality.

\paragraph{JL Properties}

There are a number of different ways to classify JL matrices.
The ones we will use are based on the above mentioned moment norm:

\begin{definition}[JL-moment property]
   We say a distribution over random matrices $M\in \R^{m\times d}$ has the $(\eps,\delta,p)$-JL-moment property, when
   \begin{align}
      \|\|Mx\|_2^2 - 1\|_p \le \eps \delta^{1/p}
   \end{align}
   for all $x\in \R^d$, $\|x\|=1$.
\end{definition}
By Markov's inequality, the JL-moment-property implies
$E\|Mx\|_2 = \|x\|_2$ and
that taking $m=O(\eps^{-2}\log1/\delta)$ suffices to have $\Pr[|\|Mx\|_2-\|x\|_2|>\eps]<\delta$ for any $x\in \R^d$.
(This is sometimes known as the Distributional-JL property.)

Note that the Strong JL Moment Property implies the $(\eps,\delta,\log1/\delta)$-JL Moment Property, since then $\eps \delta^{1/p} = \eps/e$.

\paragraph{Notation for various matrix products}
In this article we will be combining matrices in a number of different ways.
Here we introduce the definitions and properties thereof which we will use in the later sections.

\paragraph{Tensor product}
Given two matrices $A\in\R^{m\times n}$ and $B\in\R^{k\times\ell}$
we define the ``tensor product'' (or Kronecker product), $A\otimes B\in\R^{mk\times n\ell}$ as
   \begin{align}
      A\otimes B = \begin{bmatrix} A_{1,1} B & \cdots & A_{1,n}B \\ \vdots & \ddots & \vdots \\ A_{m,1} B & \cdots & A_{m,n} B \end{bmatrix}.
   \end{align}
In particular of two vectors: $x\otimes y = [x_1 y_1, x_1 y_2, \dots, x_n y_n]^T$.
Notice this equals the flattened outer product $xy^T$.
Taking the tensor-product of a vector with itself, we get the tensor-powers:
$x^{\otimes k} = \underbrace{x \otimes \dots \otimes x}_{k \text{ times}}$.

The Kronecker product has the useful ``mixed product property'', when the sizes match up:
$
   (A\otimes B)(C\otimes D) = (AC)\otimes (BD)
$.
We note in particular the vector variants
$(I\otimes B)(x\otimes y) = x \otimes By$,
$\langle x\otimes y, z\otimes t\rangle = \langle x, z\rangle\langle y, t\rangle$
and
$\langle x^{\otimes k}, y^{\otimes k}\rangle = \langle x, y\rangle^k$.

\vspace{1em}

A related operation is the ``direct sum'' for vectors: $x \oplus y = \left[\begin{smallmatrix} x \\ y \end{smallmatrix}\right]$
and for matrices: $A \oplus B = \left[\begin{smallmatrix}A & 0 \\ 0 & B \end{smallmatrix}\right]$.
When the sizes match up, we have $(A \oplus B)(x\oplus y) = Ax + By$.
Also notice that if $I_k$ is the $k\times k$ identity matrix, then $I_k \otimes A = \underbrace{A\oplus\dots\oplus A}_{k \text{ times}}$.

\paragraph{Hadamard product}
Another useful vector product is the ``Hadamard product'', also sometimes known as the `element-wise product'.
We define $x \circ y = \left[ x_1 y_1, x_2 y_2, \dots, x_n y_n \right]^T$.
Taking the Hadamard product with itself gives the Hadamard-power:
$x^{\circ k} = \underbrace{x \circ \dots \circ x}_{k \text{ times}} = [x_1^k, x_2^k, \dots, x_n^k]^T$.

\paragraph{Face-splitting product}
The Face-splitting product~\cite{slyusar2003generalized} (or transposed Khatri-Rao product) is defined as the rows-by-rows tensor product:
\begin{align}
   M = M^{(1)}\bullet\dots\bullet M^{(c)}
   &= \begin{bmatrix}
      M^{(1)}_1 \otimes M^{(2)}_1 \otimes \dots \otimes M^{(c)}_1 \\
      \vdots\\
      M^{(1)}_m \otimes M^{(2)}_m \otimes \dots \otimes M^{(c)}_m
   \end{bmatrix}
   .
\end{align}
The product has the property (which follows directly from the tensor product), that $(M\bullet T)(x\otimes y) = Mx \circ Ty$.


\subsection{Auxiliary Lemmas}
In this section we state the auxiliary lemmas that we will use throughout the paper.
Most of the lemmas are already known and the proofs of the new lemmas are deferred to
\Cref{app:aux-lemmas-proof}.

\begin{lemma}[Khintchine's inequality~\cite{haagerup2007best}] \label{lem:khintchine}
    Let $p > 0$, $x\in \R^d$, and $(\sigma_i)_{i\in[d]}$
    be independent Rademacher $\pm1$ random variables.
    Then
    \begin{align}
       \norm{\sum_{i=1}^d \sigma_i x_i}{p} \,\le C_p\, \norm{x}{2}^2
       ,
    \end{align}
    where $C_p\le \sqrt2 \left(\frac{\Gamma((p+1)/2)}{\sqrt\pi}\right)^{1/p}\le\sqrt{p}$ for $p\ge 1$.
 
    One may replace $(\sigma_i)$ with an arbitrary independent sequence of random variables $(\varsigma_i)$ with $\Ep{\varsigma_i}=0$ and $\norm{\varsigma_i}p \lesssim \sqrt{p}$, and the lemma still holds up to a universal constant factor on the rhs.
\end{lemma}

We also need a version of Khintchine's inequality that work on tensors of Rademacher vectors.
\begin{lemma}[Generalized Khintchine's Inequality]\label{lem:gen-khinchine}
    Let $p \ge 1$, $c \in Z_{> 0}$, and $(\sigma^{(i)}\in\R^{d_i})_{i \in [c]}$
    be independent vectors each satisfying the Khintchine inequality $\norm{\langle\sigma^{(i)},x\rangle}p \le C_p\|x\|_2$ for any vector $x\in\R^{d_i}$.
    Let $(a_{i_1, \ldots, i_{c}} \in \R)_{i_j\in[d_j],j\in[c]}$
    be a tensor in $\R^{d_1\times\dots\times d_c}$, then
    \begin{align}
       \norm{
          \sum_{i_1\in[d_1],\dots,i_c\in[d_c]}
          \left(\prod_{j \in [c]} \sigma^{(j)}_{i_j}\right)
          a_{i_1, \ldots, i_{c}}
        }{p}
          \le C_p^c
          \left(\sum_{
                i_1\in[d_1],\dots,i_c\in[d_c]
          } a_{i_1, \ldots, i_{c}}^2\right)^{1/2}
          .
    \end{align}
    Or, considering $a\in\R^{d_1\cdots d_c}$ a vector, then simply
    $\norm{\langle \sigma^{(1)}\otimes\dots\otimes\sigma^{(c)}, a\rangle}p
    \le C_p^c\, \|a\|_2$.
\end{lemma}
This is related to Lata{\l}a's estimate for Gaussian chaoses~\cite{latala2006estimates}, but more simple in the case where $a$ is not assumed to have special structure.
Note that this implies the classical bound on the fourth moment of products of 4-wise independent hash functions~\cite{DBLP:conf/stacs/BravermanCLMO10, indyk2008declaring, DBLP:journals/jacm/PatrascuT12},
since $C_4=3^{1/4}$ for Rademachers we have $\norm{\langle \sigma^{(1)}\otimes\dots\otimes\sigma^{(c)}, x\rangle}{4}^4 \le 3^c\, \norm{x}{2}^4$ for four-independent $\sigma^{(i)}$s.
\begin{proof}
    The proof is deferred to \Cref{app:aux-lemmas-proof}
\end{proof}

A very useful result for computing the $p$-norm of a sum of random variables is the following:
\begin{lemma}[Latała's inequality, \cite{latala1997estimation}]\label{lem:latala-sup}
    If $p\ge 2$ and $X, X_1, \dots, X_n$ are iid. mean 0 random variables, then we have
    \begin{align}
       \norm{\sum_{i=1}^n X_i}p \sim \sup
       \left\{ \frac{p}s\left(\frac np\right)^{1/s}\norm{X}s \,\middle\vert \,  \max\left\{2,\frac pn\right\}\le s\le p\right\}.
    \end{align}
\end{lemma}
 
 The following lemma first appeared in \cite{hitczenko1994domination}, but the following is roughly taken from \cite{de2012decoupling}, which we also recommend for readers interested in more general versions.
 \begin{lemma}[General decoupling, \cite{de2012decoupling} Theorem 7.3.1, paraphrasing]\label{lem:gen_decoup}
    Given a sequence $X_1, \dots, X_n$ of random variables and a filtration $\mathcal F_1, \dots, \mathcal F_n$.
    Define $Y_1, \dots, Y_n$ such that
    \begin{enumerate}
       \item $\Epcond{Y_i}{\mathcal F_{i-1}} = \Epcond{X_i}{\mathcal F_{i-1}}$ for all $i$.
       \item The sequence $(Y_i)_i$ is conditionally independent given $X_1, \dots, X_n$.
       \item $\Epcond{Y_i}{\mathcal F_{i-1}} = \Epcond{Y_i}{X_1, \dots, X_n}$ for all $i$.
    \end{enumerate}
    Then for all $p\ge 1$,
    \begin{align}
       \norm{\sum_i X_i}p \lesssim \norm{\sum_i Y_i}p
    \end{align}
 \end{lemma}
 
The next lemma is a type of Rosenthal inequality, but which mixes large and small moments in a careful way.
It bears similarity to the one sided bound in~\cite{boucheron2013concentration} (Theorem 15.10) derived from the Efron Stein inequality, and the literature has many similar bounds, but we still include a proof here based on first principles. 
\begin{lemma}[Rosenthal-type inequality]\label{lem:rosenthal-variant}
    Let $p \ge 2$ and $X_0, \ldots, X_{k - 1}$ be independent non-negative random variables
    with $p$-moment, then
    \[
       \norm{\sum_{i \in [k]} (X_i - \Ep{X_i})}{p}
          \lesssim \sqrt{p}\sqrt{\sum_{i \in [k]} \Ep{X_i}}\norm{\max_{i \in [k]} X_i}{p}^{1/2}
             + p\norm{\max_{i \in [k]} X_i}{p}
    \]
\end{lemma}
\begin{proof}
    The proof is deferred to \Cref{app:aux-lemmas-proof}.
\end{proof}

\section{Constructions}\label{sec:constructions}

The moral of this section is that, if $M$ is a JL matrix, and $M'$ is an independent copy, then for many classic constructions of such matrices,
$M \bullet M'$ is also a JL matrix.
This is important for tensor sketching, because we have the identity $(M\bullet M')(x\otimes y) = Mx \circ M'y$.
That is, matrices on this form can be efficiently applied to tensors.

Note that it is unfortunately not possible to give a general guarantee from the JL property.
To see this, note that it doesn't destroy the JL property of a matrix with $m$ rows to add $m$ more rows with all 0s.
If we append $m$ such rows to $M$ and prepend as many rows to $M'$ then $M\bullet M'=0$.

If the rows of $M$ are independent the guarantee does hold.
This follows essentially from our analysis in \cref{sec:normconst}.
Unfortunately many interesting JL distributions don't have independent rows.
One particular example is the Fast Johnson Lindenstrauss algorithm by Ailon and Chazelle~\cite{ailon2006approximate}.
They considered a matrix $SHD\in\R^{m\times d}$ where $S\in R^{m\times d}$ is a sampling matrix with one 1 per row; $H\in\R^{d\times d}$ is the Hadamard matrix defined as $\left[\begin{smallmatrix}1&1\\1&-1\end{smallmatrix}\right]^{\otimes k}$ when $d=2^k$; and $D\in\R^{d\times d}$ is a diagonal matrix with $\pm1$ entries.
Because all three matrices allow fast matrix multiplication, Fast JL allows reducing the dimension from $d$ to $m$ in time $O(m+d\log d)$.

We show that nevertheless, if $SHD$ is a Fast JL matrix and $S'H'D'$ is an independent copy, then $SHD \bullet S'H'D'$ is a JL matrix as well.
Our analysis is based on Nelson~\cite{jelnotes} and loses a factor of $\log1/\delta$ in $m$ compared to the best known method~\cite{krahmer2011new} based on the Restricted Isometry Property,
but other than that it gets the precise $(\log1/\delta)^c$ behavior that we expect from fully independent (and slow) JL matrices.

\vspace{.5em}

The constructions in this section all follow the ``direct composition'' paradigm of the original tensor sketch.
For this reason they all incur exponential dependencies in $c$.
In the next section we will reduce this to a linear dependency, but we will do so by combining the constructions of this section parameterised with $c=2$.

\subsection{Matrices with independent and identical rows}\label{sec:normconst}

The classic JL matrix, $M\in\R^{m\times d}$, is a dense matrix with independent Sub Gaussian entries, such as Gaussians or Rademachers ($\pm1$ with even chance.)
The formal definition of such a variable is that $\|X\|_p \lesssim \sqrt{p}$.

In this section we analyse our construction on a slightly more general family of matrices.
In particular we consider $M$ constructed as follows:

\begin{theorem}\label{thm:indeprows}
   Let $\eps,\delta\in[0,1]$ and let $c\ge 1$ be some integer.
   Let $T\in\R^{m\times d}$ be a matrix with iid. rows $T_1, \dots, T_m \in \R^d$ such that $\Ep{(T_1x)^2}=\|x\|_2^2$ and $\norm{T_1x}p \le \sqrt{a p} \|x\|_2$ for some $a>0$ and all $p\ge 2$.
   Let $M = T^{(1)}\bullet\dots\bullet T^{(c)}$ where $T^{(1)}, \dots, T^{(c)}$ are independent copies of $T$.
   Then there exists a constant $K>0$ such that $M$ has the \sjlmp{\eps}{\delta} given
   \begin{align}
      m \ge K \left[(4a)^{2c} \eps^{-2} \log1/\delta
      + (2ae)^{c} \eps^{-1} (\log1/\delta)^{c}\right]
      .
   \end{align}
\end{theorem}

\paragraph{Remark}
In the particular case of rows with iid. Rademachers we get $a=\sqrt{3}/4$
(the Khintchine constant from \cref{lem:khintchine})
and
$$
   m =O\(
   3^{c} \eps^{-2} \log1/\delta\,
   + \eps^{-1} (2.36\, \log1/\delta)^{c}
   \).
   $$
The same constants also holds for standard Gaussian random variables.
In the case of constant $c=O(1)$ this is simply
$O\left(\eps^{-2} \log1/\delta + \eps^{-1} (\log1/\delta)^{c}\right)$.
This improves upon the parallel work in~\cite{sarlos_neurips} which gets $m=\Omega(\eps^{-2}\log1/\eps\delta + \eps^{-1}(\log1/\eps\delta)^c)$ in this range.

If we are only interested in the JL property, and not the strong version, it is possible to get $\(\frac{\log1/\delta}{c}\)^c$ in the right term at the cost of an extra factor $e^c$ on the left term.

The proof is based on a generalized Khintchine inequality for tensor products (\cref{lem:gen-khinchine}), as well as the following consequence of a strong result by Latała (\cref{lem:latala-sup}):
\begin{corollary}\label{cor:latala-cor2}
   Let $p\ge2, C>0$ and $\alpha\ge 1$.
   Let $(X_i)_{i\in[n]}$ be iid. mean 0 random variables such that $\norm{X_i}p\sim (C p)^\alpha$,
   then
   $\norm{\sum_i X_i}p \lesssim C^\alpha\max\{2^\alpha\sqrt{pn},\, (ep)^\alpha\}$.
\end{corollary}
We prove this in \cref{proof:latala-cor2}, but for now we focus on the proof of~\cref{thm:indeprows}.
\begin{proof}
   Without loss of generalization we may assume $\norm{x}2=1$.
   We notice that $\norm{\|Mx\|_2^2 - 1}p \le \norm{\tfrac1m\sum_i (M_i x)^2 - 1}p$
   is the mean of iid. random variables.
   Call these $Z_i = (M_ix)^2-1$.
   Then $EZ_i = 0$
   and $\norm{Z_i}p = \norm{(M_ix)^2-1}p \lesssim \norm{(M_ix)^2}p = \norm{M_ix}{2p}^2$ by symmetrization (see e.g.~\cite{jelnotes}).
   Now by the assumption $\norm{T_1x}p \le \sqrt{a p} \|x\|_2 = \sqrt{ap}$, and by \cref{lem:gen-khinchine}, we get that
   $\norm{M_ix}{p} = \norm{(T^{(1)}_i\otimes\dots\otimes T_i^{(c)})x}p \le (ap)^{c/2}$,
   and so $\norm{Z_i}p \le (2ap)^{c}$ for all $i\in[m]$.

   We now use \cref{cor:latala-cor2} which implies
   \begin{align}
      \norm{\frac1m\sum_i Z_i}p
      \le
      L \max\left\{(4a)^{c}\sqrt{p/m}, \, (2eap)^{c}/m\right\},
   \end{align}
   for some constant $L$.
   Taking $m = (Le)^2\max\{(4a)^{2c} \eps^{-2} \log1/\delta
   , (2ae)^{c} \eps^{-1} (\log1/\delta)^{c}\}$
   we get
   \begin{align}
      \norm{\frac1m\sum_i Z_i}p
      &\le \max\left\{
         \frac{\eps}{e}\sqrt{\frac{p}{\log1/\delta}}, \,
         \frac{\eps}{e^2}\left(\frac{p}{\log1/\delta}\right)^{c}
   \right\}
    \le \frac{\eps}{e}\sqrt{\frac{p}{\log1/\delta}},
   \end{align}
   where we used $c\ge1$ and $p\le\log1/\delta$.
\end{proof}

We show in~\cref{appendix:lowerbound-subgaussian} that the analysis is optimal up to constant factors.
This is harder than showing the upper bound, but for the particular case of $c=2$ the following simple argument gives the intuition:

Assume $M$ and $T$ are iid. Gaussian matrices and $x=e_1^{\otimes 2}$ were a simple tensor with a single 1 entry.
Then $|\|(M\bullet T)x\|_2^2-\|x\|_2^2| = |\|Mx'\circ Tx'\|_2^2-1| \sim |(gg')^2-1|$ for $g,g'\in R$ iid. Gaussians.
Now $\Pr[(gg')^2 > (1+\epsilon)] \approx \exp(-\min(\epsilon,\sqrt{\epsilon}))$,
thus requiring $m = \Omega(\epsilon^{-2}\log1/\delta + \epsilon^{-1}(\log1/\delta)^2)$
matching the upper bound.

\subsection{Fast Constructions}\label{sec:fastconst}

The construction with independent rows can be applied to order $c$ tensors in the time it takes to do $c$ matrix-vector multiplications.
The issue is that those each take time $md$.
For large $m$ and $d$ we would like to get this closer to the size of the input.
In this section we analyse an approach that takes just $m+d\log d$ time per matrix multiplication.

As mentioned we will analyse $SHD \bullet S'H'D' \bullet \dots \in \R^{m\times d^c}$ where
$S\in R^{m\times d}$ is a sampling matrix with one 1 per row; $H\in\R^{d\times d}$ is the Hadamard matrix defined as $\left[\begin{smallmatrix}1&1\\1&-1\end{smallmatrix}\right]^{\otimes k}$ when $d=2^k$; and $D\in\R^{d\times d}$ is a diagonal matrix with $\pm1$ entries.
Each of these allow fast matrix multiplication, which immediately gives that $SHD$ can be applied fast.

What is more surprising is that, given the diagonal of $D$ is a tensor product of shorter Rademacher vectors, the $SHD$ construction is particularly applicable to tensor sketching.
For an example of this, see \cref{fig:ftjlt} below.

\begin{figure}[h]
   \begin{mdframed}[backgroundcolor=yellow!5, roundcorner=20pt, innerbottommargin=15pt, innertopmargin=0]
   \begin{align}
      &SHD(x\otimes y)
      \\
      &\quad=
      \begin{bmatrix}
         1 & 0 & 0 & 0 \\
         0 & 0 & 1 & 0 \\
         0 & 1 & 0 & 0
      \end{bmatrix}
      \begin{bmatrix}
         1 & 1 & 1 & 1 \\
         1 & -1 & 1 & -1 \\
         1 & 1 & -1 & -1 \\
         1 & -1 & -1 & 1
      \end{bmatrix}
      \begin{bmatrix}
         \sigma_1 \rho_1 & 0 & 0 & 0 \\
         0 & \sigma_1 \rho_2 & 0 & 0 \\
         0 & 0 & \sigma_2 \rho_1 & 0 \\
         0 & 0 & 0 & \sigma_2 \rho_2 \\
      \end{bmatrix}
      \begin{bmatrix}
         x_1y_1 \\
         x_2y_1 \\
         x_1y_2 \\
         x_2y_2
      \end{bmatrix}
      \\[5pt]
      &\quad=
      \left(
      \begin{bmatrix}
         1 & 0 \\
         0 & 1 \\
         1 & 0
      \end{bmatrix}
      \bullet
      \begin{bmatrix}
         1 & 0 \\
         1 & 0 \\
         0 & 1
      \end{bmatrix}
      \right)
      \left(
      \begin{bmatrix}
         1 & 1 \\
         1 & -1
      \end{bmatrix}
      \otimes
      \begin{bmatrix}
         1 & 1 \\
         1 & -1
      \end{bmatrix}
      \right)
      \left(
      \begin{bmatrix}
         \sigma_1 & 0 \\
         0 & \sigma_2 \\
      \end{bmatrix}
      \otimes
      \begin{bmatrix}
         \rho_1 & 0 \\
         0 & \rho_2 \\
      \end{bmatrix}
      \right)
      \left(
      \begin{bmatrix}
         x_1 \\
         x_2
      \end{bmatrix}
      \otimes
      \begin{bmatrix}
         y_1 \\
         y_2
      \end{bmatrix}
      \right)
      \\[5pt]
      &\quad=
      \left(
      \begin{bmatrix}
         1 & 0 \\
         0 & 1 \\
         1 & 0
      \end{bmatrix}
      \bullet
      \begin{bmatrix}
         1 & 0 \\
         1 & 0 \\
         0 & 1
      \end{bmatrix}
      \right)
      \left(
      \begin{bmatrix}
         1 & 1 \\
         1 & -1
      \end{bmatrix}
      \begin{bmatrix}
         \sigma_1 & 0 \\
         0 & \sigma_2 \\
      \end{bmatrix}
      \begin{bmatrix}
         x_1 \\
         x_2
      \end{bmatrix}
      \,\otimes\,
      \begin{bmatrix}
         1 & 1 \\
         1 & -1
      \end{bmatrix}
      \begin{bmatrix}
         \rho_1 & 0 \\
         0 & \rho_2 \\
      \end{bmatrix}
      \begin{bmatrix}
         y_1 \\
         y_2
      \end{bmatrix}
      \right)
      \\[5pt]
      &\quad=
      \begin{bmatrix}
         1 & 0 \\
         0 & 1 \\
         1 & 0
      \end{bmatrix}
      \begin{bmatrix}
         1 & 1 \\
         1 & -1
      \end{bmatrix}
      \begin{bmatrix}
         \sigma_1 & 0 \\
         0 & \sigma_2 \\
      \end{bmatrix}
      \begin{bmatrix}
         x_1 \\
         x_2
      \end{bmatrix}
      \,\circ\,
      \begin{bmatrix}
         1 & 0 \\
         1 & 0 \\
         0 & 1
      \end{bmatrix}
      \begin{bmatrix}
         1 & 1 \\
         1 & -1
      \end{bmatrix}
      \begin{bmatrix}
         \rho_1 & 0 \\
         0 & \rho_2 \\
      \end{bmatrix}
      \begin{bmatrix}
         y_1 \\
         y_2
      \end{bmatrix}
      .
   \end{align}
   \end{mdframed}
   \caption{An example of splitting up Hadamard matrices and tensor-diagonal matrices for the
   Fast Tensor Johnson Lindenstrauss transformation using the mixed product property of tensor products.}
   \label{fig:ftjlt}
\end{figure}

The properties we will use are the following:
\begin{enumerate}
   \item If $S$ and $S'$ are iid. sampling matrices with independent rows, then $S\bullet S'$ is also a sampling matrix with independent rows.
      This follows because sampling a random value in $[d^2]$ can be decomposed as sampling $i_1,i_2\in[d]$ and then taking $i=i_1d+i_2$.
      This is exactly what happens when tensoring a row of $S$ with a row of $S'$.
   \item Since the Hadamard matrix $H$ of size $d$ equals $\left[\begin{smallmatrix}1&1\\1&-1\end{smallmatrix}\right]^{\otimes k}$ for some $k$ (remember we assume $d$ is a power of 2), we naturally have that $H\otimes H$ is the Hadamard matrix of size $d^2$.
   \item If $D$ and $D'$ are diagonal matrices with respectively $\sigma\in\R^d$ and $\rho\in\R^d$ on their diagonals, then it is easy to check that $D\otimes D'$ is a diagonal matrix with $\sigma\otimes\rho$ on its diagonal.
\end{enumerate}

From these facts we have that $SHD\bullet S'H'D' = S'' H D''$ which is exactly a Fast JL construction, except for $D$ now having a simple tensor on its diagonal instead of a fully independent Rademacher vector.
The theorem below will show that, up to some extra log factors, this is not a problem.

\begin{theorem}[Fast Tensor Johnson Lindenstrauss]\label{thm:tensorfastjl}
   Let $c \in \Z_{> 0}$, and $(D^{(i)})_{i \in [c]} \in \prod_{i \in [c]} \R^{d_i \times d_i}$
   be independent diagonal matrices with independent Rademacher variables. Define
   $d = \prod_{i \in [c]} d_i$ and $D = \bigotimes_{i \in [c]} D_i \in \R^{d\times d}$.
   Let $S \in \R^{m \times d}$ be an independent sampling matrix
   which samples exactly one coordinate per row.
   Let $x \in \R^d$ be any vector with $\|x\|_2=1$ and $p \ge 1$, then
   \begin{align}
      \norm{\tfrac{1}{m}\norm{SHDx}{2}^2 - 1}{p} \lesssim \sqrt{\frac{pq^c}{m}} + \frac{pq^c}{m} ,
   \end{align}
   where $q=\max\{p,\log m\}$.
   Setting $m = \epsilon^{-2}\log1/\delta(K' \log1/\epsilon\delta)^c$ for some universal constant $K'$,
   we get that $\frac1{\sqrt{m}}SHD$ satisfies the Strong JL Moment Property (\cref{defn:strong-jl-moment}).
\end{theorem}
Note that setting $c=1$, the analysis is very similar to the Fast Johnson Lindenstrauss analysis in~\cite{cohen2016optimal, jelnotes}.

\begin{proof}[Proof of \Cref{thm:tensorfastjl}]
    For every $i \in [m]$ we let $S_i$ be the random variable that says which
    coordinate the $i$'th row of $S$ samples, and we define the random variable
    $Z_i = M_i x = H_{S_i} D x$. We note that since the variables $(S_i)_{i \in [m]}$
    are independent then the variables $(Z_i)_{i \in [m]}$ are conditionally
    independent given $D$, that is, if we fix $D$ then $(Z_i)_{i \in [m]}$ are
    independent.

    We use \Cref{lem:rosenthal-variant}, the triangle inequality, and Cauchy-Schwartz to get that
    \begin{align}\label{eq:FJLT-proof-eq}
        &\norm{\tfrac{1}{m}\sum_{i \in [m]} Z_i^2 - 1}{p}
            \\&\quad= \norm{\Epcond{\left(\tfrac{1}{m}\sum_{i \in [m]} Z_i^2 - 1\right)^p}{D}^{1/p}}{p} \label{eq:tfjl-cond}
              \\&\quad\lesssim \frac{1}{m}\norm{\sqrt{p}\, \Epcond{\left(\max_{i \in [m]} Z_i^2 \right)^{p}}{D}^{1/(2p)} \sqrt{\sum_{i \in [m]}\Epcond{Z_i^2}{D}} 
                   + p\, \Epcond{\left(\max_{i \in [m]} Z_i^2 \right)^{p}}{D}^{1/p}}{p} \label{eq:tfjl-lem5}
            \\&\quad\le \frac{\sqrt{p}}{m} \norm{\Epcond{\left(\max_{i \in [m]} Z_i^2 \right)^{p}}{D}^{1/(2p)} \sqrt{\sum_{i \in [m]}\Epcond{Z_i^2}{D}}}{p}
                + \frac{p}{m}\norm{\max_{i \in [m]} Z_i^2}{p} \label{eq:tfjl-tri}
            \\&\quad\le \frac{\sqrt{p}}{m} \norm{\max_{i \in [m]} Z_i^2}{p}^{1/2} \norm{\sum_{i \in [m]}\Epcond{Z_i^2}{D}}{p}^{1/2}
                + \frac{p}{m}\norm{\max_{i \in [m]} Z_i^2}{p} \label{eq:tfjl-cs}
        \; .
    \end{align}
    Here \cref{eq:tfjl-cond} follows from the definition of the $p$-norm as well as the law of total expectation, \cref{eq:tfjl-lem5} follows from \cref{lem:rosenthal-variant}, \cref{eq:tfjl-tri} is the triangle inequality, and \cref{eq:tfjl-cs} uses Cauchy Schwarz: $\norm{AB}p \le \norm{A}{2p}\norm{B}{2p}$ as well as a few manipulations of norms and powers.
 
    By orthogonality of $H$ we have $\norm{HDx}{2}^2 = d\norm{x}{2}^2$ independent of $D$.
    Hence 
    \[
        \sum_{i \in [m]} \Epcond{Z_i^2}{D} = \sum_{i \in [m]} \norm{x}{2}^2 = m
        \; .
    \]
    To bound $\norm{\max_{i \in [m]} Z_i^2}{p}$ we first use \cref{lem:gen-khinchine} to show
    \begin{align}
        \norm{Z_i^2}{p}
            = \norm{H_{S_i} D x}{2p}^2
            = \norm{D x}{2p}^2
            \le p^{k} \norm{x}{2}^2
        \; .
    \end{align}
    We then bound the maximum using a sufficiently high powered sum:
    \begin{align}
        \norm{\max_{i \in [m]} Z_i^2}{p}
            \le \norm{\max_{i \in [m]} Z_i^2}{q}
            \le \left(\sum_{i \in [m]} \norm{Z_i^2}{q}^q\right)^{1/q}
            \le m^{1/q} q^{c} \norm{x}{2}^2
            \le e q^{c}
        \; ,
    \end{align}
    where the last inequality follows from $q \ge \log m$.
    This gives us that
    \[
        \norm{\tfrac{1}{m}\sum_{i \in [m]} Z_i^2 - \norm{x}{2}^2}{p}
            \lesssim \sqrt{\frac{pq^c}{m}} + \frac{pq^c}{m}
        \; ,
    \]
    which finishes the first part of the proof.

    To show the Strong JL Moment Property we choose $q=2e \log m/\delta$
    and $m$ such that $m=K\eps^{-2}(\log 1/\delta)q^c \lesssim \eps^{-2}(\log 1/\delta)(K'\log 1/\eps\delta)^c$ for some universal constants $K$ and $K'$.
    Then
    \begin{align}
       \sqrt{\frac{pq^c}{m}} + \frac{pq^c}{m}
       \le \eps \left(\sqrt{\frac{p}{K\log 1/\delta}} + \frac{\eps p}{K\log 1/\delta}\right)
       \le \frac{\eps}{e} \sqrt{\frac{p}{\log 1/\delta}}
    \end{align}
    for $p\le \log1/\delta$.
\end{proof}

\section{The High Probability Tensor Sketch}

In the previous section we assumed that a tensor sketch had to look like $M = M^{(1)} \bullet M^{(2)} \bullet \dots \bullet M^{(c)}$ for some matrices $M^{(i)}$,
such that $M(x\otimes y\otimes\dots) = M^{(1)}x \circ M^{(2)}y \circ\dots$.
This worked well when $c$ was constant, but as we saw it produced a matrix with a number of rows exponential in $c$.

In this section we instead consider a ``sketch and reduce'' approach.
We'd like to have
\begin{align}
   M(x\otimes y \otimes z \dots)
   =
   M^{(1)}(x \otimes M^{(2)}(y\otimes M^{(3)}(z \otimes \dots)))
\end{align}
where we assume the vectors are already sufficiently reduced.
If $M^{(i)}$ are good tensor sketches for $c=2$ this combination should be fast and succinct.
The matrix $M$ that expands as above on tensors is
\begin{align}
   M = M^{(c)}(M^{(c - 1)} \otimes I_{d})(M^{(c-2)} \otimes I_{d^2})\ldots(M^{(1)} \otimes I_{d^{c - 1}}),
\end{align}
or more formally $M=Q^{(c)}$ where $Q^{(0)} = 1 \in \R$ and recursively $Q^{(i)} = M^{(i)}(Q^{(i - 1)} \otimes I_d) \in \R^{m\times d^i}$.

In this section we will prove \cref{thm:gen-sketch} using this construction.
We do this in two steps: 1) We show that each of $M^{(c - i)} \otimes I_{d^i}$ has the \sjlmp\eps\delta and 2) That the product of such matrices have it.

The first step follows rather easily from the following lemma:
\begin{lemma}\label{lem:jl-sum}
   Let $\eps \in (0, 1)$, $\delta > 0$ and $p\ge 0$.
   If $P \in \R^{m_1 \times d_1}$ and $Q \in \R^{m_2 \times d_2}$ are two random matrices (not necessarily independent)
   such that $\norm{\norm{Px}2^2-1}p$ and $\norm{\norm{Qx}2^2-1}p$ are both upper bounded by $\kappa$ for all $x\in\R^{d_1}$ (resp. $x\in\R^{d_2}$), $\norm{x}2=1$;
   then $\norm{(P \oplus Q)x}p \le \kappa$
\end{lemma}
\begin{proof}
   Let $x \in \R^{d_1 + d_2}$ and choose $y \in \R^{d_1}$ and $z \in \R^{d_2}$
   such that $x = y \oplus z$.
   Now using the triangle inequality and, we get that
   \begin{align}
      \norm{\norm{(P \oplus Q)x}{2}^2 - \norm{x}{2}^2}{p}
         &=\norm{\norm{Py}{2}^2 + \norm{Pz}{2}^2 - \norm{y}{2}^2 - \norm{z}{2}^2}{p}
         \\&\le \norm{\norm{Py}{2}^2 - \norm{y}{2}^2}{p} + \norm{\norm{Qz}{2}^2 - \norm{z}{2}^2}{p}
         \\&\le \kappa\norm{y}{2}^2 + \kappa\norm{z}{2}^2
         \\&= \kappa\norm{y\oplus z}{2}^2
         ,
   \end{align}
   which is what we want, since $y\oplus z=x$.
\end{proof}

An easy consequence of this lemma is that for any matrix $Q$ with the
\sjlmp{\eps}{\delta}, then $I_\ell \otimes Q$ and $Q\otimes I_\ell$ have the \sjlmp{\eps}{\delta} too.
This follows simply from
$
   I_\ell \otimes Q = \underbrace{Q \oplus Q \oplus \ldots \oplus Q}_{\ell\text{ times}}
$ and the fact that you can obtain $Q \otimes I_\ell$ by reordering the rows and columns of $I_\ell \otimes Q$,
which does not change the JL moment property since it corresponds to permutations of the input or output vectors.


We continue to show a ``product lemma'' for JL matrices, to complement the one we just proved for sums.
\begin{lemma}\label{lem:jl-product}
    There exists a universal constant $L$,
    such that for any $\eps, \delta \in [0,1]$,
    if
    $
        M^{(1)} \in \R^{d_2 \times d_1}, \ldots,
        M^{(c)} \in \R^{d_{c+1} \times d_c}
    $
    are independent random matrices with the Strong
    $(\eps/(L\sqrt{k}), \delta)$-JL Moment Property,
    then the matrix $M = M^{(c)} \cdots M^{(1)}$
    has the \sjlmp{\eps}{\delta}.
\end{lemma}
\begin{proof}
    Let $x \in \R^d$ be an arbitrary, fixed unit vector, and fix $1 < p \le \log(1/\delta)$.
    We define $X_i = \norm{M^{(i)} \cdots M^{(1)}x}{2}^2$ and
    $Y_i = X_i - X_{i - 1}$ for every $i \in [k]$. By telescoping we then have that
    $X_k - 1 = \sum_{i \in [k]} Y_i$.

    We will prove by induction on $k\in[c]$ that
    \begin{align}\label{eq:jl-moment-prod-induc}
        \norm{\sum_{i \in [k]} Y_i}{p}
            \le \frac{\eps}{e}\sqrt{\frac{t}{\log(1/\delta)}}
            \le 1
            \quad\text{for all $k\in[c]$.}
    \end{align}
    Intuitively this should be true, since $Y_i$'s are a sub-Gaussian Martingale difference sequence, however the lack of a moment generating function makes it more tricky to show.
    Note that since $\sum_{i \in [c]} Y_i = X_c-1 = \norm{Mx}2^2-1$, this is exactly the statement that $M$ has the \sjlmp{\eps}{\delta}.

    Induction start:
    For $i = 1$ we have $Y_1 = X_1 - X_0 = \norm{M^{(1)}x}{2}^2 - \norm{x}{2}^2$.
    Now as $M^{(1)}$ has the \sjlmp{\eps/(L\sqrt{k})}{\delta} we get that
    \begin{align}
        \norm{\sum_{i \in [1]} Y_i}{p} =
        \norm{\norm{M^{(1)}x}{2}^2 - 1}{p}
            \le \frac{\eps}{e L \sqrt{k}} \sqrt{\frac{p}{\log(1/\delta)}}
            \le \frac{\eps}{e}\sqrt{\frac{p}{\log(1/\delta)}}
        .
    \end{align}

    For the induction step we introduce $(T^{(i)})_{i \in [k]}$ as
    independent copies of $(M^{(i)})_{i \in [k]}$ and define
    $
        Z_i
        = \norm{T^{(i)} M^{(i - 1)} \ldots M^{(1)}x}{2}^2 - \norm{M^{(i - 1)} \ldots M^{(1)}x}{2}^2
    $
    for every $i\in[k]$.
    We can verify the following three properties:
    \begin{enumerate}
        \item $\Prpcond{Z_k > t}{(M^{(j)})_{j \in [k - 1]}} = \Prpcond{Y_k > t}{(M^{(j)})_{j \in [k - 1]}}$ for every $t \in \R, k \in [c]$.
        \item The sequence $(Z_i)_{i \in [c]}$ is conditionally independent given $(M^{(i)})_{i \in [c]}$.
        \item $\Prpcond{Z_k > t}{(M^{(i)})_{i \in [c - 1]}} = \Prpcond{Z_k > t}{(M^{(i)})_{i \in [c]}}$ for every $t \in \R, k \in [c]$.
    \end{enumerate}
    This means we can use \Cref{lem:gen_decoup} to get
    \begin{align}\label{eq:jl-prod-decoup}
       \norm{\sum_{i \in [k]} Y_i}{p} \le C_0\norm{\sum_{i \in [k]} Z_i}{p}
        \quad\text{for every $i \in [c]$}
    \end{align}
    and so it suffices to show \cref{eq:jl-moment-prod-induc} on $\sum_{i \in [k]} Z_i$ which is somewhat more well behaved than $\sum_{i \in [k]} Y_i$.

    Now assume that \eqref{eq:jl-moment-prod-induc} is true for $i - 1$. Using \eqref{eq:jl-prod-decoup}
    we get that $\norm{X_i - 1}{p} = \norm{\sum_{j \in [i]} Y_j}{p} \le C_0 \norm{\sum_{j \in [i]} Z_j}{p}$.
    By using that $(T^{(j)})_{j \in [i]}$ has
    the Strong $(\eps/(L\sqrt{k}), \delta)$-JL Moment Property together with Khintchine's inequality
    (\Cref{lem:khintchine}), we get that
    \begin{align}
        \norm{\sum_{j \in [i]} Z_j}{p}
            &= \norm{\Epcond{\left( \sum_{j \in [i]} Z_j \right)^p}{(M^{(j)})_{j \in [i]}}^{1/p}}{p}
            \\&\le C_1 \norm{\frac{\eps}{e L \sqrt{k}}\sqrt{\frac{p}{\log(1/\delta)}}\sqrt{\sum_{j \in [i]} X_j^2 }}{p}
            \\&\le C_1 \frac{\eps}{e} \sqrt{\frac{p}{\log(1/\delta)}} \,\frac{1}{L \sqrt{k}} \sqrt{\sum_{j \in [i]} \norm{X_j}{p}^2}
        ,
    \end{align}
    where the last inequality follows from the triangle inequality.
    Using the triangle inequality and \eqref{eq:jl-moment-prod-induc} we have
    $\norm{X_i}{p}
            \le 1 + \norm{X_k - 1}{p}
            = 1 + \norm{\sum_{i\in[k]}Y_i}{p}
            \le 2
            $ by the induction hypothesis.
    Setting $L = 2C_0C_1$ we get that
    \begin{align}
        \norm{\sum_{i \in [k]} Y_i}{p}
            &\le \frac{\eps}{e} \sqrt{\frac{p}{\log(1/\delta)}} \,\frac{C_0 C_1}{L \sqrt{k}} \sqrt{\sum_{i \in [k]} \norm{X_i}{p}^2}
            \\&\le \frac{\eps}{e}\sqrt{\frac{p}{\log(1/\delta)}} \,\frac{C_0 C_1}{L \sqrt{k}} \,2\sqrt{k}
            \\&\le \frac{\eps}{e}\sqrt{\frac{p}{\log(1/\delta)}}
        \; ,
    \end{align}
    which finishes the induction.
\end{proof}

Combining the two lemmas finally gives \cref{thm:gen-sketch}.


\section{Oblivious Subspace Embedding}\label{sec:techniques}
In this last section of paper we show how to get an Oblivious Subspace Embedding
using our results. First we formally define Oblivious Subspace Embedding.

\begin{definition}[$\eps$-Subspace embedding]
   $M\in \R^{k\times D}$ is a subspace embedding for $\Lambda\subseteq \R^D$
   if for any $x\in\Lambda$,
   \begin{align}
      |\|Mx\|_2 - \|x\|_2| \le \eps\|x\|_2.
   \end{align}
\end{definition}

\begin{definition}[$(\lambda,\eps,\delta)$-Oblivious Subspace Embedding]
   \label{defn:ose}
   A distribution, $\mathcal M$, over $R^{m\times D}$ matrices is a $(D,\lambda)$-Oblivious Subspace Embedding if for any linear subspace, $\Lambda \subseteq \R^{D}$, of dimension $\lambda$,
   $M\sim\mathcal M$ is an $\eps$-subspace embedding for $\Lambda$ with probability at least $1-\delta$.
\end{definition}

We show two different ways to obtain an Oblivious Subspace Embedding. The first approach
is using Approximate Matrix Multiplication, which is formally defined below, and the
approach is using a standard net-argument.

\begin{definition}[$(\eps,\delta)$-Approximate Matrix Multiplication]
   We say a distribution over random matrices $M\in \R^{k\times d}$ has the $(\eps,\delta)$-Approximate Matrix Multiplication property if for any matrices $A,B$ with proper dimensions,
   \begin{align}
      \|\|A^TM^TMB-A^TB\|_F\|_{p} \le \eps\delta^{1/p}\|A\|_F\|B\|_F.
   \end{align}
\end{definition}

We will show that the JL-moment property implies Approximate Matrix Multiplication.
To do this we first show that the JL-moment property implies concentration on the inner product
of two vectors.
\begin{lemma}[Two vector JL-moment property]
   \label{lem:bamm}
   For any $x,y\in \R^d$, if $S$ has the $(\eps,\delta)$-JL-moment-property, then also
   \begin{align}
      \| (Sx)^T(Sy) - x^Ty \|_p \le \eps\delta^{1/p}\|x\|_2\|y\|_2
      \label{eq:better-mat-mul}
   \end{align}
\end{lemma}
\begin{proof}
   We can assume by linearity of the norms that $\|x\|_2=\|y\|_2=1$.
   We then use that $\|x-y\|^2_2 = \|x\|^2_2 + \|y\|^2_2 - 2x^Ty$
   and $\|x+y\|^2_2 = \|x\|^2_2 + \|y\|^2_2 + 2x^Ty$
   such that $x^Ty = (\|x+y\|^2_2-\|x-y\|^2_2)/4$.

   Plugging this into the left hand side of \cref{eq:better-mat-mul} gives
   \begin{align}
      \| (Sx)^T(Sy) - x^Ty \|_p
        &=
        \big\| \|Sx+Sy\|^2_2 - \|x+y\|^2_2 - \|Sx-Sy\|^2_2 + \|x-y\|^2_2 \big\|_p/4
      \\&\le
      (\big\| \|S(x+y)\|^2_2 - \|x+y\|^2_2\big\|_p + \big\|\|S(x-y)\|^2_2 - \|x-y\|^2_2 \big\|_p)/4
      \\&\le
      \eps\delta^{1/p} (\|x+y\|^2_2+\|x-y\|^2_2)/4 \quad\text{(JL property)}
      \\&=
      \eps\delta^{1/p} (\|x\|^2_2+\|y\|^2_2)/2
      \\&=
      \eps\delta^{1/p}
      .
   \end{align}
\end{proof}

Now we can prove the following lemma which was first proved in~\cite{kane2014sparser} theorem 6.2, see also~\cite{DBLP:journals/fttcs/Woodruff14}.
\begin{lemma}[JL implies Approximate Matrix Multiplication]
   \label{lemma:amm}
   Any distribution that has the $(\eps,\delta,p)$-JL-moment-property
   has the $(\eps,\delta,p)$-Approximate Matrix Multiplication property.
\end{lemma}
We note that~\cite{DBLP:journals/fttcs/Woodruff14} has a factor of $3$ on $\eps$, but this can be removed with \Cref{lem:bamm}.
We give a short proof for completion:
\begin{proof}
   \begin{align}
      \norm{\norm{A^TM^TMB-A^TB}F}p
      &= \Big\|\sum_{i,j} (A_i M^T M B_j^T - A_i B_j^T)^2\Big\|_{p/2}^{1/2}
    \\&\le \sqrt{\sum_{i,j} \norm{A_i M^T M B_j^T - A_i B_j^T}{p}^2}
    \\&\le \sqrt{\sum_{i,j} (\eps\delta^{1/p}\norm{A_i}2 \norm{B_j}2)^2}
    \\&= \eps\delta^{1/p}\norm{A}F\norm{B}F.
   \end{align}
\end{proof}

\begin{lemma}
   \label{lemma:ose}
   Any distribution that has the $(\eps/\lambda,\delta)$-JL-moment-property
   is a $(\lambda,\eps)$-oblivious subspace embedding.
\end{lemma}
\begin{proof}
   Let $U\in \R^{\lambda\times m}$ be orthonormal such that $U^TU=I$,
   it then suffices (by~\cite{DBLP:journals/fttcs/Woodruff14}) to show $\|U^T M^T M U - I\|\le\eps$.
   From \Cref{lemma:amm} we have that
   $\|U^T M^T M U - I\| \le \eps\delta^{1/p}\|U\|_F^2 = \eps\delta^{1/p}\lambda$.
\end{proof}

The second way to get an Oblivious Subspace Embedding is using a standard net-argument.
\begin{lemma}
   \label{lemma:ose2}
   There is a $C>0$, such that
   any distribution that has the $(\eps,\delta e^{-C \lambda})$-JL-moment-property
   is a $(\lambda,\eps)$-oblivious subspace embedding.
\end{lemma}
\begin{proof}
   For any $\lambda$-dimensional subspace, $\Lambda$, there exists an $\eps$-net $T\subseteq \Lambda\cap S^{d-1}$ of size $C^d$ such that if $M$ preserves the norm of every $x\in T$ then $M$ preserves all of $\Lambda$ up to $1+\eps$.
   See~\cite{DBLP:journals/fttcs/Woodruff14} for details.
\end{proof}

\subsection{Combination lemmas}

It is also possible to show lemmas similar to \cref{lem:jl-sum} and \cref{lem:jl-product} for Oblivious Subspace Embeddings.
Combined with an OSE analysis of JL, merging that of \cite{cohen2016optimal} and \cref{thm:tensorfastjl} one gets another way to achieve a fast OSE tensor sketch with linear dependence on $\lambda$.

\begin{lemma}[A Direct Sum is an OSE]
   Let $M\in\R^{m\times d}$ be an $(\eps,\delta/b)$-oblivious subspace embedding,
   then $M\otimes I_b\in\R^{mb\times db}$ is an $(\eps,\delta)$-oblivious subspace embedding.
\end{lemma}
Intuitively,
let $A_{[i]}$ denote the $i$th block of $A = [A_{[1]},\dots,A_{[b]}]^T$.
By a union bound $M$ is a subspace embedding for each block.
We then intuitively have
\begin{align}
   \norm{(M\otimes I_b)Ax}{2}^2
   = \sum_{i\in[b]}\norm{MA_{[i]}x}{2}^2
 = (1\pm\eps)\sum_{i\in[b]}\norm{A_{[i]}x}{2}^2
 = (1\pm\eps)\norm{Ax}{2}^2.
\end{align}

\begin{proof}
We can make this more formal, using that $\norm{X}{} = \sup_x |x^T X x|$ for symmetric $X$ and $x$ a unit vector.
This means that the subspace property is equivalent with that
      $|x^T(A_{[i]}^TM^TMA_{[i]}-A_{[i]}^TA_{[i]})x| \le \eps\, x^TA_{[i]}^TA_{[i]}x$ for all $i\in[b]$ and vectors $x$.

Hence
   \begin{align}
      \norm{A^T(M\otimes I_{b})^T(M\otimes I_{b})A-A^TA}{}
      &= \sup_{x} \left|x^T(A^T(M\otimes I_{b})^T(M\otimes I_{b})A-A^TA)x\right|
    \\&= \sup_{x} \left|\sum_{i\in[b]} x^T(A_{[i]}^TM^TMA_{[i]}-A_{[i]}^TA_{[i]})x\right|
    \\&\le \sup_{x} \sum_{i\in[b]} \left|x^T(A_{[i]}^TM^TMA_{[i]}-A_{[i]}^TA_{[i]})x\right|
    \\&\le \sup_{x} \sum_{i\in[b]} \eps\, x^T A_{[i]}^T A_{[i]} x
    \\& = \eps \sup_{x} x^T A^T A x
    \\&= \eps\norm{A^TA}{}
    \\&= \eps\norm{A}{}^2.
   \end{align}
   This implies the non-symmetric case
   \begin{align}
      \norm{A^T(M\otimes I_{b})^T(M\otimes I_{b})B-A^TB}{} \le 2\eps\norm{A}{}\norm{B}{}
   \end{align}
   by an argument in~\cite{cohen2016optimal}.
   (Basically taking a single space large enough to contain both $A$ and $B$.)
\end{proof}

Note that this matches what we have for the JL-moment-property except that it requires a union bound, rather than just individual moment bounds.

\begin{lemma}[A Product is an OSE]
   Let $M^{(1)}, \dots, M^{(c)}$ have the Moment-OSE property,
   in other words assume that
   $\norm{\norm{A^TM^{(i)T}M^{(i)}A-I}{}}p \le \eps\delta^{1/p}$ for some $p$ and all $i$, when $A^TA=I$.

   Then
   \begin{align}
      \norm{\norm{A^TM^{(c)T}\cdots M^{(1)T}M^{(1)}\cdots M^{(c)}A-I}{}}p \le (1+\eps\delta^{1/p})^c-1
      .
   \end{align}
\end{lemma}
Note that for small $\eps$ or $\delta$ this is $\approx n\eps\delta^{1/p}$.
\begin{proof}
   We show this for the case $c=2$.
   \begin{align}
      \norm{\norm{A^T(MT)^T(MT)A-I}{}}p
      &\le \norm{\norm{A^TT^TM^TMTA-A^TT^TTA}{}}p+\norm{\norm{A^TT^TTA-I}{}}p
      \\&\le \eps\delta^{1/p}\norm{\norm{A^TT^TTA}{}}p+\eps\delta^{1/p}
      \\&\le \eps\delta^{1/p}(\norm{\norm{A^TT^TTA-I}{}}p+1)+\eps\delta^{1/p}
      \\&\le (1+\eps\delta^{1/p})^2-1.
   \end{align}
   By induction these results can be extended to give $(1+\eps\delta^{1/p})^c-1\approx c\eps\delta^{1/p}$ for composing $n$ subspace embeddings.

   Alternatively we can show the same thing directly without normalization:
   \begin{align}
      \norm{\norm{A^T(MT)^T(MT)B-A^TB}{}}p
      &\le \norm{\norm{A^TT^TM^TMTB-A^TT^TTB}{}}p+\norm{\norm{A^TT^TTB-A^TB}{}}p
      \\&\le \eps\delta^{1/p}\norm{\norm{TA}{}\norm{TB}{}}p+\eps\delta^{1/p}\norm{A}{}\norm{B}{}
      \\&\le \eps\delta^{1/p}\norm{\norm{T}{}^2}p\norm{A}{}\norm{B}{}+\eps\delta^{1/p}\norm{A}{}\norm{B}{}
      \\&\le \eps\delta^{1/p}(\norm{\norm{T^TT-I}{}}p+1)\norm{A}{}\norm{B}{}+\eps\delta^{1/p}\norm{A}{}\norm{B}{}
      \\&\le ((1+\eps\delta^{1/p})^2-1)\norm{A}{}\norm{B}{}.
   \end{align}
\end{proof}

It is likely that a ``strong OSE moment'' property would be able to reduce that to $\approx \sqrt{n}\eps\delta^{1/p}$, similar to the JL case.
However this seems to be outside of current random matrix techniques.

\section{Acknowledgements}
We would like to thank Anders Aamand for discussions of a previous version of this article.

Also Jelani Nelson, Ninh Pham and the anonymous reviewers for helpful comments and corrections.

\bibliographystyle{plainurl}
\bibliography{paper}

\begin{thebibliography}{10}

\bibitem{ailon2006approximate}
Nir Ailon and Bernard Chazelle.
\newblock Approximate nearest neighbors and the fast johnson-lindenstrauss
  transform.
\newblock In {\em Proceedings of the thirty-eighth annual ACM symposium on
  Theory of computing}, pages 557--563. ACM, 2006.

\bibitem{sarlos_neurips}
Anonymous.
\newblock High probability bounds for sketching tensor products.
\newblock In double blind review, 2019.

\bibitem{DBLP:conf/nips/AvronNW14}
Haim Avron, Huy~L. Nguyen, and David~P. Woodruff.
\newblock Subspace embeddings for the polynomial kernel.
\newblock In {\em Advances in Neural Information Processing Systems 27: Annual
  Conference on Neural Information Processing Systems 2014, December 8-13 2014,
  Montreal, Quebec, Canada}, pages 2258--2266, 2014.
\newblock URL:
  \url{http://papers.nips.cc/paper/5240-subspace-embeddings-for-the-polynomial-kernel}.

\bibitem{boucheron2013concentration}
St{\'e}phane Boucheron, G{\'a}bor Lugosi, and Pascal Massart.
\newblock {\em Concentration inequalities: A nonasymptotic theory of
  independence}.
\newblock Oxford university press, 2013.

\bibitem{DBLP:conf/stacs/BravermanCLMO10}
Vladimir Braverman, Kai{-}Min Chung, Zhenming Liu, Michael Mitzenmacher, and
  Rafail Ostrovsky.
\newblock {AMS} without 4-wise independence on product domains.
\newblock In {\em 27th International Symposium on Theoretical Aspects of
  Computer Science, {STACS} 2010, March 4-6, 2010, Nancy, France}, pages
  119--130, 2010.
\newblock URL: \url{https://doi.org/10.4230/LIPIcs.STACS.2010.2449}, \href
  {http://dx.doi.org/10.4230/LIPIcs.STACS.2010.2449}
  {\path{doi:10.4230/LIPIcs.STACS.2010.2449}}.

\bibitem{charikar2002similarity}
Moses~S Charikar.
\newblock Similarity estimation techniques from rounding algorithms.
\newblock In {\em Proceedings of the thiry-fourth annual ACM symposium on
  Theory of computing}, pages 380--388. ACM, 2002.

\bibitem{cohen2018simple}
Michael~B Cohen, TS~Jayram, and Jelani Nelson.
\newblock Simple analyses of the sparse johnson-lindenstrauss transform.
\newblock In {\em 1st Symposium on Simplicity in Algorithms (SOSA 2018)}.
  Schloss Dagstuhl-Leibniz-Zentrum fuer Informatik, 2018.

\bibitem{cohen2016optimal}
Michael~B Cohen, Jelani Nelson, and David~P Woodruff.
\newblock Optimal approximate matrix product in terms of stable rank.
\newblock In {\em 43rd International Colloquium on Automata, Languages, and
  Programming (ICALP 2016)}. Schloss Dagstuhl-Leibniz-Zentrum fuer Informatik,
  2016.

\bibitem{de2012decoupling}
Victor De~la Pena and Evarist Gin{\'e}.
\newblock {\em Decoupling: from dependence to independence}.
\newblock Springer Science \& Business Media, 2012.

\bibitem{gao2016compact}
Yang Gao, Oscar Beijbom, Ning Zhang, and Trevor Darrell.
\newblock Compact bilinear pooling.
\newblock In {\em Proceedings of the IEEE conference on computer vision and
  pattern recognition}, pages 317--326, 2016.

\bibitem{guest2017success}
Olivia Guest and Bradley~C Love.
\newblock What the success of brain imaging implies about the neural code.
\newblock {\em Elife}, 6:e21397, 2017.

\bibitem{haagerup2007best}
Uffe Haagerup and Magdalena Musat.
\newblock On the best constants in noncommutative khintchine-type inequalities.
\newblock {\em Journal of Functional Analysis}, 250(2):588--624, 2007.

\bibitem{hitczenko1993domination}
Pawe{\l} Hitczenko.
\newblock Domination inequality for martingale transforms of a rademacher
  sequence.
\newblock {\em Israel Journal of Mathematics}, 84(1-2):161--178, 1993.

\bibitem{hitczenko1994domination}
Pawel Hitczenko.
\newblock On a domination of sums of random variables by sums of conditionally
  independent ones.
\newblock {\em The Annals of Probability}, pages 453--468, 1994.

\bibitem{indyk2008declaring}
Piotr Indyk and Andrew McGregor.
\newblock Declaring independence via the sketching of sketches.
\newblock In {\em Proceedings of the nineteenth annual ACM-SIAM symposium on
  Discrete algorithms}, pages 737--745. Society for Industrial and Applied
  Mathematics, 2008.

\bibitem{johnson1984extensions}
William~B Johnson and Joram Lindenstrauss.
\newblock Extensions of lipschitz mappings into a hilbert space.
\newblock {\em Contemporary mathematics}, 26(189-206):1, 1984.

\bibitem{kane2014sparser}
Daniel~M Kane and Jelani Nelson.
\newblock Sparser johnson-lindenstrauss transforms.
\newblock {\em Journal of the ACM (JACM)}, 61(1):4, 2014.

\bibitem{shdpk}
Michael Kapralov, Rasmus Pagh, Ameya Velingker, David~P. Woodruff, and Amir
  Zandieh.
\newblock Sketching high-degree polynomial kernels.
\newblock In review, 2019.

\bibitem{krahmer2011new}
Felix Krahmer and Rachel Ward.
\newblock New and improved johnson--lindenstrauss embeddings via the restricted
  isometry property.
\newblock {\em SIAM Journal on Mathematical Analysis}, 43(3):1269--1281, 2011.

\bibitem{latala1997estimation}
Rafa{\l} Lata{\l}a.
\newblock Estimation of moments of sums of independent real random variables.
\newblock {\em The Annals of Probability}, 25(3):1502--1513, 1997.

\bibitem{latala2006estimates}
Rafa{\l} Lata{\l}a et~al.
\newblock Estimates of moments and tails of gaussian chaoses.
\newblock {\em The Annals of Probability}, 34(6):2315--2331, 2006.

\bibitem{DBLP:journals/corr/LeSS14}
Quoc Le, Tam{\'a}s Sarl{\'o}s, and Alex Smola.
\newblock Fastfood-approximating kernel expansions in loglinear time.
\newblock In {\em Proceedings of the international conference on machine
  learning}, volume~85, 2013.

\bibitem{lin2015bilinear}
Tsung-Yu Lin, Aruni RoyChowdhury, and Subhransu Maji.
\newblock Bilinear cnn models for fine-grained visual recognition.
\newblock In {\em Proceedings of the IEEE international conference on computer
  vision}, pages 1449--1457, 2015.

\bibitem{musco2017recursive}
Cameron Musco and Christopher Musco.
\newblock Recursive sampling for the nystrom method.
\newblock In {\em Advances in Neural Information Processing Systems}, pages
  3833--3845, 2017.

\bibitem{jelnotes}
Jelani Nelson.
\newblock Lecture notes on algorithms for big data, October 2015.

\bibitem{nguyen2015tensor}
Nam~H Nguyen, Petros Drineas, and Trac~D Tran.
\newblock Tensor sparsification via a bound on the spectral norm of random
  tensors.
\newblock {\em Information and Inference: A Journal of the IMA}, 4(3):195--229,
  2015.

\bibitem{DBLP:journals/jacm/PatrascuT12}
Mihai Patrascu and Mikkel Thorup.
\newblock The power of simple tabulation hashing.
\newblock {\em J. {ACM}}, 59(3):14:1--14:50, 2012.
\newblock URL: \url{https://doi.org/10.1145/2220357.2220361}, \href
  {http://dx.doi.org/10.1145/2220357.2220361}
  {\path{doi:10.1145/2220357.2220361}}.

\bibitem{pham2013fast}
Ninh Pham and Rasmus Pagh.
\newblock Fast and scalable polynomial kernels via explicit feature maps.
\newblock In {\em Proceedings of the 19th ACM SIGKDD international conference
  on Knowledge discovery and data mining}, pages 239--247. ACM, 2013.

\bibitem{rahimi2008random}
Ali Rahimi and Benjamin Recht.
\newblock Random features for large-scale kernel machines.
\newblock In {\em Advances in neural information processing systems}, pages
  1177--1184, 2008.

\bibitem{slyusar2003generalized}
VI~Slyusar.
\newblock Generalized face-products of matrices in models of digital antenna
  arrays with nonidentical channels.
\newblock {\em RADIO ELECTRONICS AND COMMUNICATIONS SYSTEMS C/C OF
  IZVESTIIA-VYSSHIE UCHEBNYE ZAVEDENIIA RADIOELEKTRONIKA}, 46(10):9--17, 2003.

\bibitem{valiant2015finding}
Gregory Valiant.
\newblock Finding correlations in subquadratic time, with applications to
  learning parities and the closest pair problem.
\newblock {\em Journal of the ACM (JACM)}, 62(2):13, 2015.

\bibitem{DBLP:conf/icml/WeinbergerDLSA09}
Kilian~Q. Weinberger, Anirban Dasgupta, John Langford, Alexander~J. Smola, and
  Josh Attenberg.
\newblock Feature hashing for large scale multitask learning.
\newblock In {\em Proceedings of the 26th Annual International Conference on
  Machine Learning, {ICML} 2009, Montreal, Quebec, Canada, June 14-18, 2009},
  pages 1113--1120, 2009.
\newblock URL: \url{https://doi.org/10.1145/1553374.1553516}, \href
  {http://dx.doi.org/10.1145/1553374.1553516}
  {\path{doi:10.1145/1553374.1553516}}.

\bibitem{DBLP:journals/fttcs/Woodruff14}
David~P. Woodruff.
\newblock Sketching as a tool for numerical linear algebra.
\newblock {\em Foundations and Trends in Theoretical Computer Science},
  10(1-2):1--157, 2014.
\newblock URL: \url{https://doi.org/10.1561/0400000060}, \href
  {http://dx.doi.org/10.1561/0400000060} {\path{doi:10.1561/0400000060}}.

\end{thebibliography}

\appendix


\section{Proof of Auxiliary Lemmas}\label{app:aux-lemmas-proof}

{
\renewcommand{\thelemma}{\ref{lem:gen-khinchine}}
\begin{lemma}[Generalized Khintchine's Inequality]
    Let $p \ge 1$, $c \in Z_{> 0}$, and $(\sigma^{(i)}\in\R^{d_i})_{i \in [c]}$
    be independent vectors each satisfying the Khintchine inequality $\norm{\langle\sigma^{(i)},x\rangle}p \le C_p\|x\|_2$ for any vector $x\in\R^{d_i}$.
    Let $(a_{i_1, \ldots, i_{c}} \in \R)_{i_j\in[d_j],j\in[c]}$
    be a tensor in $\R^{d_1\times\dots\times d_c}$, then
    \begin{align}
       \norm{
          \sum_{i_1\in[d_1],\dots,i_c\in[d_c]}
          \left(\prod_{j \in [c]} \sigma^{(j)}_{i_j}\right)
          a_{i_1, \ldots, i_{c}}
        }{p}
          \le C_p^c
          \left(\sum_{
                i_1\in[d_1],\dots,i_c\in[d_c]
          } a_{i_1, \ldots, i_{c}}^2\right)^{1/2}
          .
    \end{align}
    Or, considering $a\in\R^{d_1\cdots d_c}$ a vector, then simply
    $\norm{\langle \sigma^{(1)}\otimes\dots\otimes\sigma^{(c)}, a\rangle}p
    \le C_p^c\, \|a\|_2$.
\end{lemma}
\addtocounter{lemma}{-1}
}
\begin{proof}
    
   The proof will be by induction on $c$. For $c = 1$ then the result is by assumption.
   So assume that the result is true for every value up to $c-1$.
   Let $B_{i_1,\dots,i_{c-1}} = \sum_{i_{c} \in [d_{c}]} \sigma^{(c)}_{i_{c}} a_{i_1, \ldots, i_{c}}$.
   We then pull it out of the left hand term in the theorem:
   \begin{align}
      \norm{
         \sum_{i_1\in[d_1],\dots,i_c\in[d_c]}
         \left(\prod_{j \in [c]} \sigma^{(j)}_{i_j}\right)
         a_{i_1, \ldots, i_{c}}
      }{p}
      &=
         \norm{
            \sum_{i_1\in[d_1],\dots,i_{c-1}\in[d_{c-1}]}
            \left(\prod_{j \in [c - 1]} \sigma^{(j)}_{i_j}\right)
            B_{i_1,\dots,i_{c-1}}
         }{p}
      \\&\le
         C_p^{c-1}
         \norm{
            \left(
               \sum_{i_1\in[d_1],\dots,i_{c-1}\in[d_{c-1}]}
               B_{i_1,\dots,i_{c-1}}^2
            \right)^{1/2}
         }{p}
         \label{eq:gen-khin-ih}
       \\&=
         C_p^{c-1}
         \norm{
               \sum_{i_1\in[d_1],\dots,i_{c-1}\in[d_{c-1}]}
               B_{i_1,\dots,i_{c-1}}^2
            }{p/2}^{1/2}
      \\&\le
         C_p^{c-1}
         \left(
            \sum_{i_1\in[d_1],\dots,i_{c-1}\in[d_{c-1}]}
         \norm{B_{i_1,\dots,i_{c-1}}^2}{p/2} \right)^{1/2}
         \label{eq:gen-khin-triangle}
      \\&=
         C_p^{c-1}
         \left(
            \sum_{i_1\in[d_1],\dots,i_{c-1}\in[d_{c-1}]}
         \norm{B_{i_1,\dots,i_{c-1}}}{p}^2 \right)^{1/2}
         .
   \end{align}
   Here \cref{eq:gen-khin-ih} is the inductive hypothesis and \cref{eq:gen-khin-triangle} is the triangle inequality.
   Now $\norm{B_{i_1,\dots,i_{c-1}}}{p}^2 \le C_p^2\sum_{i_c\in[d_c]}a_{i_1,\dots,i_c}^2$ by Khintchine's inequality, which finishes the induction step and hence the proof.
\end{proof}

{
\renewcommand{\thelemma}{\ref{lem:rosenthal-variant}}
\begin{lemma}[Rosenthal-type inequality]
    Let $p \ge 2$ and $X_0, \ldots, X_{k - 1}$ be independent non-negative random variables
    with $p$-moment, then
    \[
       \norm{\sum_{i \in [k]} (X_i - \Ep{X_i})}{p}
          \lesssim \sqrt{p}\sqrt{\sum_{i \in [k]} \Ep{X_i}}\norm{\max_{i \in [k]} X_i}{p}^{1/2}
             + p\norm{\max_{i \in [k]} X_i}{p}
    \]
\end{lemma}
\addtocounter{lemma}{-1}
}
\begin{proof}
    \begin{align}
       \norm{\sum_{i \in [k]} (X_i - \Ep{X_i})}{p}
          &\lesssim
             \norm{\sum_{i \in [k]} \sigma_i X_i}{p}
          \quad\text{(Symmetrization)}
          \\&\lesssim
             \sqrt{p}\norm{\sum_{i \in [k]} X_i^2}{p/2}^{1/2}
          \quad\text{(Khintchine's inequality)}
          \\&\le
          \sqrt{p}\norm{\(\max_{i\in[k]} X_i\)\sum_{i \in [k]} X_i}{p/2}^{1/2}
          \quad\text{(Non-negativity)}
          \\&\le
             \sqrt{p}\norm{\max X_i}{p}^{1/2}\norm{\sum_{i \in [k]} X_i}{p}^{1/2}
          \quad\text{(Cauchy)}
          \\&\le
          \sqrt{p}\norm{\max X_i}{p}^{1/2}\left(\sqrt{\sum_{i \in [k]} \Ep{X_i}}
          + \norm{\sum_{i \in [k]} (X_i - \Ep{X_i})}{p}^{1/2}\right)
          .
    \end{align}
    Now let $C = \norm{\sum_{i \in [k]} (X_i - \Ep{X_i})}{p}^{1/2}$,
    $B = \sqrt{\sum_{i \in [k]} \Ep{X_i}}$,
    and $A=\sqrt{p}\norm{\max X_i}{p}^{1/2}$.
    then we have shown $C^2 \le A(B+C)$.
    That implies $C$ is smaller than the largest of the roots of the quadratic.
    Solving this quadratic inequality gives
    $C^2 \lesssim AB+A^2$
    which is the result.
 \end{proof}

\section{Upper and Lower Bounds for Sub-Gaussians}

The following simple corollaries will be used for both upper and lower bounds:
\begin{corollary}\label{cor:latala-cor}
   Let $2n\ge p\ge2, C>0$ and $\alpha\ge 1$.
   Let $(X_i)_{i\in[n]}$ be iid. mean 0 random variables such that $\norm{X_i}p\sim (C p)^\alpha$,
   then
   \begin{align}
      \norm{\sum_i X_i}p \sim C^\alpha\max\{2^\alpha\sqrt{pn},\, (n/p)^{1/p}p^{\alpha}\}.
   \end{align}
\end{corollary}
\begin{proof}
   We show that the expression $f(s)=\frac{p}s\left(\frac np\right)^{1/s}s^\alpha$ in \cref{lem:latala-sup} is maximized either by minimizing or maximizing the parameter $s$.
   It suffices to show that the derivative $\frac{df(s)}{ds}\frac{p}{s} = \frac{-p}{s^{3-\alpha}}(\frac np)^{1/s}\left((1-\alpha)s+\log\frac np\right)$ is non-decreasing in $s$, which follows from $(n/p)^{1/s}s^{\alpha-3}$ being positive.
   Hence we either take $s=\max\{2,n/p\}$ or $s=p$, but from the assumption $2n\ge p$ that gives the corollary.
\end{proof}
On the upper bound side we can simplify things a bit.
\begingroup
\def\thecorollary{\ref{cor:latala-cor2}}
\begin{corollary}
   Let $p\ge2, C>0$ and $\alpha\ge 1$.
   Let $(X_i)_{i\in[n]}$ be iid. mean 0 random variables such that $\norm{X_i}p\sim (C p)^\alpha$,
   then
   \begin{align}
      \norm{\sum_i X_i}p \lesssim C^\alpha\max\{2^\alpha\sqrt{pn},\, (ep)^\alpha\}.
   \end{align}
\end{corollary}
\addtocounter{corollary}{-1}
\endgroup
\begin{proof}\label{proof:latala-cor2}
   The requirement on the upper bound disappears since we can always take the supremum over a larger range without decreasing the result.
   We get rid of the term $(n/p)^{1/p}\le n^{1/p}$ with the following argument:
   Assume $n^{1/p}p^\alpha$ dominates the maximum.
   Then $2^\alpha\sqrt{pn}\le n^{1/p}p^\alpha
   \implies n^{1/p}
   \le ((p/2)^\alpha p^{-1/2})^{\frac{2}{p-2}}
   $.
   One can check that
   $p^{\frac{-1}{p-2}}\le 1$ and
   $(p/2)^{\frac{2}{p-2}}$ is decreasing for $p\ge2$ and so $n^{1/p}$ is bounded by
   $\lim_{p\to2} (p/2)^{\frac{2\alpha}{p-2}}
   = e^\alpha$.
\end{proof}

For the lower bound we will also use the following result by Hitczenko, which provides an improvement on Khintchine for Rademacher random variables.

\begin{lemma}[Sharp bound on Rademacher sums~\cite{hitczenko1993domination}]\label{lem:hitczenko}
   Let $\sigma\in\{-1,1\}^n$ be a random Rademacher sequence and let $a\in\R^n$ be an arbitrary real vector, then
   \begin{align}
      \norm{\langle a,\sigma\rangle}p \sim \sum_{i\le p}a_i + \sqrt{p}\big(\sum_{i>p}a_i^2\big)^{1/2}
   \end{align}
\end{lemma}

Finally the lower bound will use the Paley-Zygmund inequality (also known as the one-sided Chebyshev inequality):
\begin{lemma}[Paley-Zygmund]\label{lem:pz}
   Let $X\ge 0$ be a real random variable with finite variance, and let $\theta\in[0,1]$, then
   \begin{align}
      \Prp{X >\theta \Ep{X}} \ge (1-\theta)^2\frac{\Ep{X}^2}{\Ep{X^2}}.
   \end{align}
\end{lemma}
A classical strategy when using Paley-Zygmund is to prove $\Ep{X}\ge2\eps$ and $\Ep{X}^2/\Ep{X^2}>4\delta$ for some $\eps,\delta>0$, and then take $\theta=1/2$ to give
$\Prp{X > \eps} > \delta$.

\subsection{Lower Bound for Sub-Gaussians}\label{appendix:lowerbound-subgaussian}

The following lower bound considers the sketching matrix consisting of the direct composition of matrices with Rademacher entries.
Note however that the assumptions on Rademachers are only used to show that the $p$-norm of a single row with a vector is $\sim\sqrt{p}$.
For this reason the same lower bound hold if the Rademacher entries are substituted for say Gaussians.

\begin{theorem}[Lower bound]\label{thm:lower}
   For some constants $C_1,C_2,B>0$, let $d,m,c\ge1$ be integers, let $\eps\in[0,1]$ and $\delta\in[0,1/16]$.
   Further assume that $d\ge\log1/\delta\ge c/B$.
   Then the following holds.

   Let $M^{(1)}, \dots, M^{(c)} \in \R^{m\times d}$ be matrices with all independent Rademacher entries
   and let $M=\tfrac1{\sqrt m} M^{(1)}\bullet\dots\bullet M^{(c)}$.
   Then there exists some unit vector $y\in\R^{d^c}$ such that if
   \begin{align}
      m < C_1 \max\left\{
            3^c \eps^{-2}\frac{\log1/\delta}c
         ,\, \eps^{-1}\left(\frac{C_2\log1/\delta}c\right)^{c}\right\}
            \quad\text{then}\quad \Pr\left[\abs{\norm{My}2^2-1}>\eps\right]>\delta.
   \end{align}
\end{theorem}
\begin{proof}
   Let $y=[1,\dots,1]^T/\sqrt{d}\in\R^d$ and let $x=y^{\otimes c}$.
   We have
   \begin{align}
      \norm{My}2^2-1
      = \frac1m\left\|M^{(1)}x \circ \dots \circ M^{(c)}x\right\|_2^2-1
      = \frac1m\sum_{j\in[m]} \big(\prod_{i\in[c]} Z_{i,j}^2\big)-1
   \end{align}
   where each $Z_{i,j}=\sum_{k\in[d]}M^{(i)}_{j,k}/\sqrt{d}$ are independent averages of $d$ independent Rademacher random variables.
   By~\cref{lem:hitczenko} we have $\norm{Z_{i,j}}p \sim \min\{\sqrt{p},\sqrt{d}\}$ which is $\sqrt{p}$ by the assumption $d\ge\log1/\delta$ as long as $p\le\log1/\delta$.
   By the expanding $Z_{i,j}^4$ into monomials and linearity of expectation we get $\|Z_{i,j}\|_4 = \frac1{\sqrt{d}}(d+3d(d-1))^{1/4} = (3-2/d)^{1/4}$.

   Now define $X_j = \prod_{i\in[c]}Z_{i,j}^2-1$, then $EX_j=0$
   and $\norm{X_j}p \ge \norm{\prod_{i\in[c]}Z_{i,j}^2}p-1 = \norm{Z_{i,j}}{2p}^{2c}-1 \ge K^c p^c$ for some $K$, assuming $p\ge 2$.
   In particular, $\norm{X_j}2\ge \norm{Z_{i,j}}{4}^{2c}-1 = (3-2/d)^{c/2}-1
   \sim 3^{c/2}$ by the assumption $d\ge c\ge1$.

   We have $\norm{\norm{My}2^2-1}p = \frac1m\norm{\sum_{j\in[m]}X_m}p$ is a sum of iid. random variables,
   so we can use \cref{cor:latala-cor} to show
   \begin{align}
            K_3\max\left\{\sqrt{3^c p/m}, (m/p)^{1/p}K_1^c p^c/m\right\}
            &\lesssim \norm{\norm{My}2^2-1}p
          \\&\lesssim K_4\max\left\{\sqrt{3^c p/m}, (m/p)^{1/p}K_2^c p^c/m\right\}
            \label{eq:upper-lower}
   \end{align}
   for some universal constants $K_1,K_2,K_3,K_4>0$.

   Assume now that
   $m < \max\left\{
      AK_3^23^c \eps^{-2}\frac{\log1/\delta}c,
      \frac{K_3}4\eps^{-1}\left(4AK_1\frac{\log1/\delta}c\right)^{c}
   \right\}$ as in the theorem.
   We take $p=4A\frac{\log1/\delta}{c}$ for some constant $A$ to be determined.
   We want to show $\norm{\norm{My}2^2-1}p \ge 2\eps$.
   For this we split into two cases depending on which term of $m<\max\{(1),(2)\}$ dominates.
   If $(1) \ge (2)$ we pick the first lower bound in \cref{eq:upper-lower} and get
   $\norm{\norm{My}2^2-1}p \ge K_3\sqrt{3^c p/m} \ge K_3 \sqrt{\frac{4\eps^2}{K_3^2}} = 2\eps$.
   Otherwise, if $(2)\ge(1)$, we pick the other lower bound and also get:
   \begin{align}
      \norm{\norm{My}2^2-1}p \ge
      K_3 (m/p)^{1/p}\frac{K_1^c p^c}m
      \ge
      \frac{K_3}2 \frac{K_1^c \left(4A\frac{\log1/\delta}{c}\right)^c}{\frac{K_3}4\eps^{-1}\left(4AK_1\frac{\log1/\delta}c\right)^{c}}
      = 2\eps,
   \end{align}
   where we used $(m/p)^{1/p}\ge e^{-1/(em)}\ge1/2$ for $m\ge 1$.
   Plugging into Paley-Zygmund (\cref{lem:pz}) we have
   \begin{align}
      \Pr\left[\abs{\norm{My}2^2-1}\ge\eps\right]
      &\ge \Pr\left[\abs{\norm{My}2^2-1}^p\ge\norm{\norm{My}2^2-1}p^p 2^{-p}\right]
    \\&\ge\frac14 \left(\frac{\norm{\norm{My}2^2-1}p}{\norm{\norm{My}2^2-1}{2p}}\right)^{2p},
       \label{eq:pz}
   \end{align}
   where we used that $p\ge 1$ so $(1-2^{-p})^2\ge1/4$.

   There are again two cases depending on which term of the upper bound in \cref{eq:upper-lower} dominates.
   If $\sqrt{3^cp/m} \ge(m/p)^{1/p}K_2^c p^c/m$ we have using the first lower bound that
   $\frac{\norm{\norm{My}2^2-1}p}{\norm{\norm{My}2^2-1}{2p}} \ge \frac{K_3}{\sqrt{2}K_4}$.
   For the alternative case, $(m/p)^{1/p}K_2^c p^c/m \ge \sqrt{3^cp/m}$, we have
   \begin{align}
      \frac{\norm{\norm{My}2^2-1}p}{\norm{\norm{My}2^2-1}{2p}}
      \ge \frac{K_3}{\sqrt2 K_4} \frac{(m/p)^{1/p}}{(m/2p)^{1/2p}}\left(\frac{K_1}{2 K_2}\right)^c
      \ge \frac{K_3}{2 K_4}\left(\frac{K_1}{2 K_2}\right)^c
   \end{align}
   where $\frac{(m/p)^{1/p}}{(m/2p)^{1/2p}}\ge e^{-1/(4em)}\ge 1/\sqrt2$ for $m\ge 1$.

   Comparing with \eqref{eq:pz} we see that it suffices to take $A \le \min\{\frac1{\log 2K_4/K_3}, \frac1{\log 2K_2/K_1}\}/32$.
   This choice also ensures that $1\le p\le \log1/\delta$ as we promised.
   Note that we may assume in \cref{eq:upper-lower} that $K_3\le K_4$ and $K_1\le K_2$.
   We then finally have
   \begin{align}
      \frac14\left(\frac{K_3}{\sqrt{2} K_4}\right)^{2p}
      \ge \frac14 \delta^{1/(4c)}
      \quad\text{and}\quad
      \frac14\left(\frac{K_3}{2K_4}\left(\frac{K_1}{2 K_2}\right)^c\right)^{2p}
      \ge \frac14 \delta^{1/(4c)+1/4},
   \end{align}
   which are both $>\delta$ for $c\ge1$ and $\delta<1/16$.
\end{proof}

\end{document}